\documentclass[aps,twocolumn,showpacs,pra,citeautoscript,amsmath,amssymb,floatfix,superscriptaddress, 10pt,nofootinbib]
{revtex4-1}
\pdfoutput=1
\usepackage{xr}
\usepackage{amsmath}
\usepackage{amssymb}
\usepackage{epsfig}
\usepackage{color}
\usepackage{amsthm}
\usepackage{hyperref}
\usepackage{appendix}
\usepackage{cancel}
\usepackage{stmaryrd}
\usepackage{soul}
\usepackage{comment}
\usepackage{empheq}

\usepackage{dsfont,bbm}

\newcommand{\be}{\begin{equation}}
\newcommand{\ee}{\end{equation}}
\newcommand{\nn}{{\mathbbm{N}}}
\newcommand{\rr}{{\mathbbm{R}}}
\newcommand{\cc}{{\mathbbm{C}}}

\newcommand{\me}{\mathrm{e}}
\newcommand{\mi}{\mathrm{i}}
\newcommand{\id}{{\mathbbm{1}}}
\newcommand{\tr}{\textup{Tr}}

\theoremstyle{definition}

\newcommand{\mpw}[1]{\textcolor{blue}{Mischa: #1}}

\newcommand{\bee}{\begin{enumerate}}
\newcommand{\eee}{\end{enumerate}}
\newcommand{\bei}{\begin{itemize}}
\newcommand{\eei}{\end{itemize}}
\newtheorem{theorem}{Theorem}

\newtheorem{lemma}[theorem]{Lemma}

\newtheorem{remark}[theorem]{Remark}

\usepackage{amsfonts}

\def\01{\{0,1\}}

\newcommand{\ket}[1]{|#1\rangle}
\newcommand{\bra}[1]{\langle#1|}

\newcommand{\proj}[1]{|#1\rangle\langle#1|}
\newcommand{\ketbra}[2]{|#1\rangle\langle#2|}

\newcommand{\alv}[1]{{\color{red}Alv: #1}}

\def\<{\langle}
\def\>{\rangle}

\newtheorem*{rep@theorem}{\rep@title}
\newcommand{\newreptheorem}[2]{%
\newenvironment{rep#1}[1]{%
 \def\rep@title{#2 \ref{##1} (restatement)}%
 \begin{rep@theorem}}%
 {\end{rep@theorem}}}
\makeatother

\newreptheorem{thm}{Theorem}
\newreptheorem{lem}{Lemma}

\begin{document}

\title{Dynamical Maps, Quantum Detailed Balance and Petz Recovery Map}

\author{\'{A}lvaro M. Alhambra}
\email{alvaro.alhambra.14@ucl.ac.uk}
\affiliation{Department of Physics and Astronomy, University College London, Gower Street, London WC1E 6BT, United Kingdom}

\author{Mischa P. Woods}
\email{mischa.woods@gmail.com}
\affiliation{Department of Physics and Astronomy, University College London, Gower Street, London WC1E 6BT, United Kingdom}
\affiliation{QuTech, Delft University of Technology, Lorentzweg 1, 2611 CJ Delft, Netherlands}

\begin{abstract}
Markovian master equations (formally known as quantum dynamical semigroups) can be used to describe the evolution of a quantum state $\rho$ when in contact with a memoryless thermal bath. This approach has had much success in describing the dynamics of real-life open quantum systems in the lab. Such dynamics increase the entropy of the state $\rho$ and the bath until both systems reach thermal equilibrium, at which point entropy production stops. Our main result is to show that the entropy production at time $t$ is bounded by the relative entropy between the original state and the state at time $2t$. The bound puts strong constraints on how quickly a state can thermalise, and we prove that the factor of $2$ is tight. The proof makes use of a key physically relevant property of these dynamical semigroups -- detailed balance, showing that this property is intimately connected with the field of recovery maps from quantum information theory. We envisage that the connections made here between the two fields will have further applications. We also use this connection to show that a similar relation can be derived when the fixed point is not thermal. 

\end{abstract}
\maketitle

\section{Introduction}

It is very often observed in nature that physical systems relax to an equilibrium state. This phenomenon, which has very evident consequences at the macroscopic scales of our everyday experience, ultimately relies on the dynamics of the microscopic components. This fact was understood in the early days of statistical mechanics, and since then a large amount of work has been produced with the aim of trying to understand how exactly physical systems reach thermal equilibrium.

Any such evolution will be ultimately generated through some reversible dynamics on a large composite system, that is effectively irreversible as seen by a smaller part of that composite system. This irreversibility means that, in a coarse-grained sense, entropy will be produced throughout the process. The entropy production can be linked to the fact that correlations between a big thermal object (a heat bath) and one smaller subsystem $S$ are increasingly harder to access, which forces the coarse-graining of the description \cite{esposito2010entropy}. Intuitively, the more irreversible a process is, the more entropy is produced, and the closer a particular system will be to equilibrium.

In this work we look at a commonly used family of quantum evolutions that model the dynamics of a system weakly coupled to a thermal bath, and show explicitly how the amount of entropy produced along a particular evolution is related to how much a state changes along that evolution. These maps were first studied by Davies \cite{davies1974} and are a quantum generalization of the classical Glauber dynamics.

In the limit of a large thermal bath, the total entropy produced by such a process is given by how much the free energy of a system decreases with time \cite{spohn1978entropy}. The free energy for a state $\rho_S(t)$ at time $t$ is defined as
\begin{equation}\label{eq:free}
F_\beta (\rho_S(t))=\text{Tr}[\hat H_S \rho_S(t)]+\frac{1}{\beta}\text{Tr}[\rho_S(t)\log{\rho_S(t)}],
\end{equation}
where $\hat H_S$ is the Hamiltonian of the subsystem of interest, and $\beta^{-1}$ is the temperature of the bath. 
Moreover, for an evolution from time $t=0$ to $t$, the total amount of Von Neumann entropy produced, the so-called \textit{Entropy Production} is given by $F_\beta (\rho_S(0))-F_\beta (\rho_S(t))= \beta \Delta  E  - \Delta S,$
with $\Delta E$, $\Delta S$ the changes in mean energy and Von Neumann entropy of the system. Due to the contractivity property of the quantum relative entropy, this quantity is non-negative and non-decreasing with $t\geq 0$.

The reason for this name is as follows. For a large thermal reservoir, small changes of energy (that is, heat transferred to the system) are proportional to changes of entropy in it, with proportionality constant $\frac{1}{\beta}$. Hence, we can identify the change in energy in the system with a change of entropy in the reservoir $\beta \Delta E  \simeq -\Delta S_{\text{bath}}$, so that the difference in free energy of the system for a time interval $\Delta t$ is equal to the total entropy generated during the interval $\Delta t$ in system and bath. Therefore, this entropy production constitutes a natural measure of the irreversibility of the process.

Our main result is Theorem \ref{thorem:Li-Winter conjecture for Davies maps}, which states that under the condition that the interaction between system and bath is time-independent, we can lower-bound the entropy production at time $t$ by the state at time $2t$. 

This sharpens some intuitive notions, namely that if not much entropy is produced during a time interval $\Delta t$, the state will not change very much during the time interval $2\Delta t$, but if it does, then a large amount of entropy must have been produced at an earlier time, namely during the time interval $\Delta t$.

Recovery maps have found many applications in quantum information theory, such as coding theorems \cite{beigi2015decoding,Brandao2016}, approximate error correction \cite{ng2010simple} or asymmetry \cite{marvian2016clocks}.
They also appear in the derivation of quantum fluctuation theorems \cite{aberg2016fully,alhambra2016fluctuating}. 

Our results, are inspired by findings in quantum information theory about recovery maps. Specifically, they are a consequence of the observation that if a dynamical map satisfies \textit{quantum detailed balance} (QDB), a property of thermodynamical processes, then this implies that the map is its own recovery map. The connection between information theory and thermodynamics goes back a long way, to the seminal work of Landauer \cite{landauer1961irreversibility} and has furthered our understanding of both significantly. Within the current surge of information-theory approaches to quantum thermodynamics (see \cite{goold2016role} for a review), our result provides another example of how ideas from one may find definite applications in the other.

We shall first introduce Davies maps, outline their properties. This is followed by the statement of the main result and a discussion on the bound itself. We finally conclude with some suggestions for open questions. 

\section{Davies maps and entropy production}\label{sec:C maps and S production Main text}

Davies maps are a particular set of quantum dynamical semigroups that describe the evolution of a system on a $d_S$ dimensional Hilbert space that is weakly interacting with a heat bath. The first rigorous derivation of their form was given in \cite{davies1974} (see \cite{alicki2007quantum,temme2013lower} for more modern treatments). As they are time-continuous quantum semigroups, their generator takes the form of a Lindbladian operator, which we define as
\begin{equation}\label{eq:lind}
\frac{\text{d} \rho_S(t)}{\text{d} t }= \mathcal{L} (\rho_S(t))+\mi \theta(\rho_S(t)),
\end{equation}
where $\mathcal{L}$ is called the \textit{Lindbladian} and $\theta(\cdot)=-[H_{\text{eff}}, \cdot]$ is called the unitary part, with $H_{\text{eff}}$ the effective Hamiltonian. The solution is a one-parameter family of CPTP maps $M_\Delta(\cdot)$, $\Delta\geq 0$ which governs the dynamics, $M_\Delta(\rho(t))= \rho(t+\Delta)$. We will not delve into the full details here, but instead highlight the important properties the canonical form of Davies maps, denoted $T_t(\cdot)$, possess:
\begin{itemize}
\item[1)] They arise from the weak system-bath coupling limit
\item [2)] They can be written in the form $T_t(\cdot)=\me^{\mi t\theta+t \mathcal{L}}(\cdot)$, with $\theta$ and $\mathcal{L}$ time independent
\item[3)]  $\theta$ and $\mathcal{L}$  commute: $\theta(\mathcal{L}(\cdot))=\mathcal{L}(\theta(\cdot))$
\item[4)] They have a thermal fixed point: $T_t(\tau_S)=\tau_S$, where $\tau_S$ is the Gibbs state of the system at temperature $T_S$.
\item [5)] Their Lindbladians and unitary part satisfy \textit{Quantum detailed balance} (QDB):
\begin{align}\label{eq:QDB}
\langle A, \mathcal{L}^\dagger (B) \rangle_{\Omega} &= \langle \mathcal{L}^\dagger (A), B \rangle_{\Omega},\\
[H_{\text{eff}},\Omega]&=0,
\end{align}
for all $A,B\in\cc^{d_S\times d_S}$, where $\mathcal{L}^\dag$ is the adjoint Lindbladian. $\Omega$ can be any quantum state. However, in the case of Davies maps, $\Omega=\tau_S$. The scalar product in Eq. \eqref{eq:QDB} is defined as
\begin{equation} \label{eq:scal}
\langle A, B \rangle_{\Omega} := \text{Tr}[\Omega^{1/2} A^\dagger \Omega^{1/2} B].
\end{equation}
This is sometimes referred to as \emph{reversibility} or KMS condition. It is stronger than 4), since it has as a consequence that $\Omega$ is the fixed point, as $\mathcal{L}(\Omega)=0$.
\end{itemize}
In Appendix \ref{App1}
 we give a more detailed account of the microscopic origin of these maps, and of the form of the weak coupling limit, property 1). In the literature, there are various different definitions of QDB which are in general not equivalent. We show in Appendix \ref{app:eqqdb}
  that for maps satisfying time translation symmetry, such as Davies maps, definition 5) is equivalent to the definition of QDB in \cite{alicki2007quantum,kossakowski1977}.

In addition to the properties above, it is sometimes assumed that:
\begin{itemize}
\item[6)] The dynamics associated with Davies maps converge to the fixed point, $\lim_{t\rightarrow\infty} T_t(\rho_S(0))=\tau_S$.
\end{itemize}
Such convergence is guaranteed if more stringent conditions are imposed on the Davies map  \cite{SpohnHerbert07,Spohn1977_converge1,Frigerio1977_converge2,Frigerio1978_converge3}. We will \textit{not} need to assume 6) here.
 
Since we wish to bound the distance from the state at time $t$ to the fixed point, we need a distance measure. For this we use the \textit{relative entropy} $D(\rho\|\sigma)=\text{Tr}[\rho (\log{\rho}-\log{\sigma})]$. This measure is meaningful since it is non-negative, zero iff $\rho=\sigma$, and is contractive under CPTP maps. For the special case that $\sigma$ is a Gibbs state, it has an interpretation in terms of a free energy,
\begin{equation}
D( \rho(t) || \tau_S) =\beta F_\beta (\rho_S(t))- \log{Z_S},
\end{equation}
where $Z_S=\tr{[e^{-\beta H_S}]}$ is the partition function of the system, which we assume is constant. 
We can thus write the entropy production in terms of a difference in relative entropy, as
\begin{align}\label{eq:entrop prod amint txt}
D( \rho(0) || \tau_S) -D( \rho(t) || \tau_S)&= \beta\left( F_\beta (\rho_S(0))-F_\beta (\rho_S(t)) \right)\end{align}

As one intuitively might expect, this entropy production only depends on the dissipative part of the dynamics, as we explain in Section 
 \ref{sec:QDB}
 of the appendix. Therefore, we will assume for simplicity that $\theta=0$ in the next Section unless stated otherwise.

If one were to change the initial state of the environment for the maximally mixed state, then the system can only exchange entropy, but not heat/energy with it. These correspond to unital maps, in which case the free energy is replaced with the entropy gain of the system alone. In that case, a lower bound on the entropy they produce in terms of the adjoint of the unital map can be found in \cite{buscemi2016approximate}.

\section{Main results}

Our main result is a tight lower bound on the change of free energy and total entropy produced, within a finite time. We start with a Lemma for Davies maps which is an initial step in its derivation:

\begin{lemma}\label{thorem:Li-Winter conjecture for Davies maps1}
All Davies maps $T_t(\cdot)$, satisfy the inequality
\be\label{Eq:proof theorem 1 eq}
 D (\rho_S(0)\|\tau_S)-D (\rho_S(t)\|\tau_S) \geq D\left(\rho_S(0) \big{\|}\tilde T_t(\rho_S(t))\right),
\ee
where $\tilde T_t(\cdot)$ is the \emph{time-reversed} map or Petz recovery map, defined as
\begin{equation}\label{eq:TQDB}
\tilde{T}_t(\cdot)= \tau^{1/2}_S T^\dagger_t \left(\tau^{-1/2}_S (\cdot) \tau^{-1/2}_S\right) \tau^{1/2}_S,
\end{equation}
with $T_t^\dag$ denoting the adjoint of $T_t$. 
\end{lemma}
\begin{proof}
See Appendix 
\ref{sec:main theorem}.
\end{proof}
Eq. \ref{eq:TQDB} proves a physically relevant particular case of an open conjecture about general quantum maps first formulated in \cite{li2014squashed}. The strongest possible version of the conjecture is known to not be true in full generality \cite{brandao2015quantum}, although it has been shown for particular sets such as unital maps \cite{buscemi2016approximate}, classical stochastic matrices\cite{li2014squashed}, catalytic thermal operations \cite{wehner2015work} and we here show it for Davies maps.  All these results relate the decrease of relative entropy with a measure of how well a given pair of states can be recovered through a particular \emph{recovery map}, and are generalizations of an early result by Petz \cite{petz1986sufficient}. For the best results up to date on general quantum maps, see \cite{wilde2015recoverability,junge2015universal,sutter2016strengthened,sutter2016multivariate}.

For Lemma \ref{thorem:Li-Winter conjecture for Davies maps1} to hold, only properties 1) and 4) are required. In addition, we find that there is a connection between property 4) and the Petz recovery map which we will now explain. 
A quantum dynamical semi-group $M_t$ which obeys QDB  has a Petz recovery map $\tilde M_t$ which is equal to the map itself $\tilde M_t= M_t$ (See Theorem 
 \ref{lem:qdb} in Appendix). 
Petz derived his famous recovery map in 1986 \cite{petz1986sufficient} while the first appearance of the detailed balance condition goes back  at least to the work of Boltzmann in 1872 \cite{Boltzmann1964} and QDB to Alicki in 1976 \cite{Alicki_DB_first}. To the best of the authors knowledge, this connection between results from the communities of quantum information theory and quantum dynamical semi-groups 
 was previously  unknown. Perhaps the closest previous work, is \cite{oftherDBdef}, which \textit{define} Detailed Balance as the property that the recovery map is equal to the map itself. Our work implies that for the special case of the Petz recovery map, the Detail Balance definition of \cite{oftherDBdef}, is equal to definition 5) which is satisfied by Davies maps.

The classical definition of detailed balance, in terms of the transition probabilities $p(j |i)$ of a classical Master equation, implies that, at equilibrium, a particular jump between energy levels $E_i \rightarrow E_j$ has the same total probability as the opposite jump $E_j \rightarrow E_i$, such that $p(j |i) \frac{e^{-\beta E_i}}{Z}=p(i |j) \frac{e^{-\beta E_j}}{Z}$. The condition in Eq. \eqref{eq:QDB} is the most natural quantum generalization of that (although as shown in \cite{temme2010chi2} different ones are also possible). In that sense, QDB can be understood as the fact that a particular thermalization process coincides with its own time-reversed map, which is defined as in Eq. \eqref{eq:TQDB} (for more details, see e.g. \cite{crooks2008quantum,Ticozzi2010}).

On the other hand, the Petz recovery map $\tilde \Gamma (\cdot)$, given a state $\sigma$ and a CPTP map $\Gamma(\cdot)$, is formally defined as \cite{petz1986sufficient,petz1988sufficiency,petz2003monotonicity} 
\begin{equation}
\tilde{\Gamma}(\cdot)= \sigma^{1/2} \Gamma^\dagger \left(\Gamma(\sigma)^{-1/2} (\cdot) \Gamma(\sigma)^{-1/2}\right) \sigma^{1/2}.
\end{equation}

This map is such that we have that iff $D(\rho || \sigma)=D(\Gamma(\rho) || \Gamma(\sigma) )$ then $\tilde \Gamma(\Gamma(\rho))=\rho$ and $\tilde \Gamma(\Gamma(\sigma))=\sigma$. It appears in quantum information theory when one tries to find the best possible way to recover data after it is processed \cite{barnum2002reversing,wilde2013quantum}. 

We can hence rewrite Lemma \ref{thorem:Li-Winter conjecture for Davies maps1} using Eq. \eqref{eq:entrop prod amint txt} as
\begin{theorem}\label{thorem:Li-Winter conjecture for Davies maps}
All Davies maps $T_t(\cdot)$, satisfy the inequality
\be\label{Eq:proof theorem}
 F_\beta (\rho_S(0))-F_\beta (\rho_S(t)) \geq \frac{1}{\beta}D\left(\rho_S(0) \big{\|} \rho_S(2t)\right).
\ee
\end{theorem}

\begin{proof}
See Appendix 
 \ref{sec:QDB}.
\end{proof}
In addition to assuming detailed balance, condition 5), we have also used condition 2). If the Lindbladian $\mathcal{L}$ is time dependent, i.e. 2) is not satisfied, Eq. \eqref{Eq:proof theorem} holds but with $\rho_S(2t)$ replaced with $T_t(\rho(t))$.

While, as mentioned at the end of Section \ref{sec:C maps and S production Main text}, entropy production is invariant under a change in the unitary part of the dynamics, it is interesting to find the Petz recovery map when $\theta$ is not set to zero. We show in Lemma 
\ref{lem:qdb unitaary}
 in the appendix, that the Petz recovery map $\tilde M_t(\cdot)$ of a map $M_t(\cdot)$ satisfying QDB and for which $\mathcal{L}$ and $\theta$ commute [property 3) of Davies maps], reverses the unitary part of the dynamics, while keeping the same dissipative part, that is
	\be 
	\tilde M_t(\cdot)=\me^{-\mi t\theta+t\mathcal{L}}(\cdot),
	\ee
and thus $\tilde T_t(T_t(\cdot)) =\me^{2t\mathcal{L}}(\cdot)$. So not only is the l.h.s. of Eq. \ref{Eq:proof theorem} invariant under a change in the unitary part of the dynamics; but so is the r.h.s.

In Fig. \ref{fig:Plot1} we show a simple example of the inequality for the case of Davies maps applied on a qutrits. Eq. \eqref{Eq:proof theorem} is tight at $t=0$ and also in the large time limit, as long as condition 6) is satisfied. In this limit, the total entropy that has been produced is equal to $\frac{1}{\beta}D(\rho(0) || \tau_S)$, which both sides of the inequality approach as $\rho_S \rightarrow \tau_S$.

\begin{figure}
\includegraphics[width=0.46\textwidth]{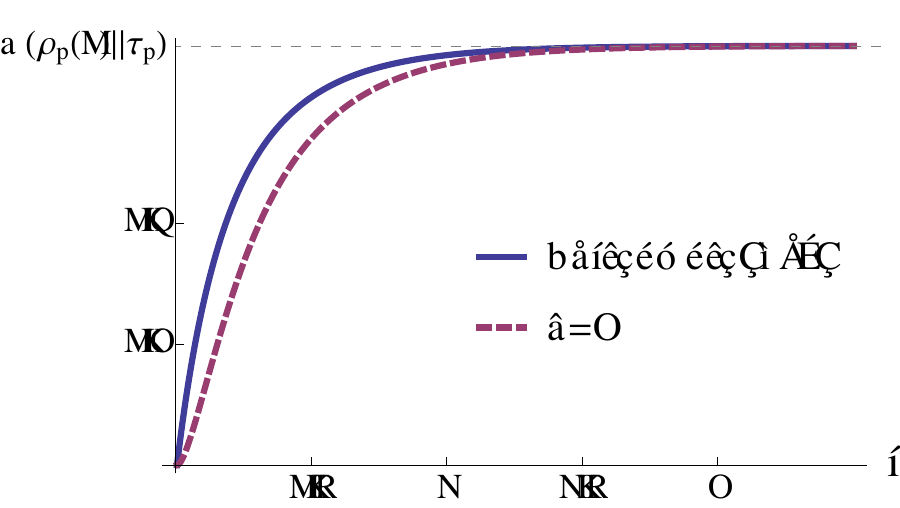}
\caption  {An example of the inequality in Theorem \ref{thorem:Li-Winter conjecture for Davies maps} for a Davies map on a qutrit given in \cite{roga2010davies}. 
The solid (blue) curve is the amount of entropy produced $\beta F_\beta(\rho(0))-\beta F_\beta(\rho(t))$ and the dashed (purple) the lower bound $D(\rho_S(0) || \rho_S(2t))$. It can be seen how the lower bound at $t=0$ starts at zero, and how for large times the two curves quickly converge to the total amount of entropy produced $ D (\rho_S(0) || \tau_S)$.  The $y$ axis is dimensionless and the $x$ axis is in units of the inverse of the coupling constant of the semigroup.} \label{fig:Plot1}
\end{figure}

On the other hand, for very short times, the lower bound becomes trivial. In particular, in Appendix 
\ref{Spohn}
 we show what both sides of the inequality tend to in the limit of infinitesimal time transformations. The entropy production becomes a \emph{rate}, and the lower bound to it approaches $0$. 

Non-trivial lower bounds on the rate of entropy production, in the form of log-Sobolev inequalities \cite{kastoryano2013quantum} can be used to derive bounds on the time it takes to converge to equilibrium for particular instances of Davies maps. Hence, given that Theorem \ref{thorem:Li-Winter conjecture for Davies maps} is completely general, and holds also for Davies maps that do not efficiently reach thermal equilibrium, the fact that the lower bound vanishes for infinitesimal times is not surprising.

Recall that the factor of $2$ in Eq. \eqref{Eq:proof theorem} is a consequence of the observation that the Petz recovery map is equal to the map itself. A natural question is then, \textit{is the factor $2$ fundamental?} We show that this is indeed the case with the following Theorem.

\begin{theorem} \label{lemmak}
	[Tightness of the entropy production bound] The largest constant $k\geq 0$ such that 
	\begin{equation}\label{eq:largek}
	F_\beta (\rho_S(0))-F_\beta (\rho_S(t)) \geq \frac{1}{\beta}D\left(\rho_S(0) \big{\|} \rho_S(k \,t)\right)
	\end{equation}
	holds for all Davies maps, is $k=2$.
\end{theorem}
\begin{proof}
	Due to Theorem \ref{thorem:Li-Winter conjecture for Davies maps}, we only need to find a simple family of Davies maps for which the violation is proven analytically for all $k>2$. See Appendix 
	\ref{app:lemmak} for proof.
\end{proof}
 See Fig. \ref{fig:Plot2} for more details. This means that Eq. \eqref{thorem:Li-Winter conjecture for Davies maps} is the strongest constraint of its kind that Davies maps obey, and it hence sets an optimal relation between how much the free energy and the systems state at a later time change during a  thermalization process.

\begin{figure}
\includegraphics[width=0.46\textwidth]{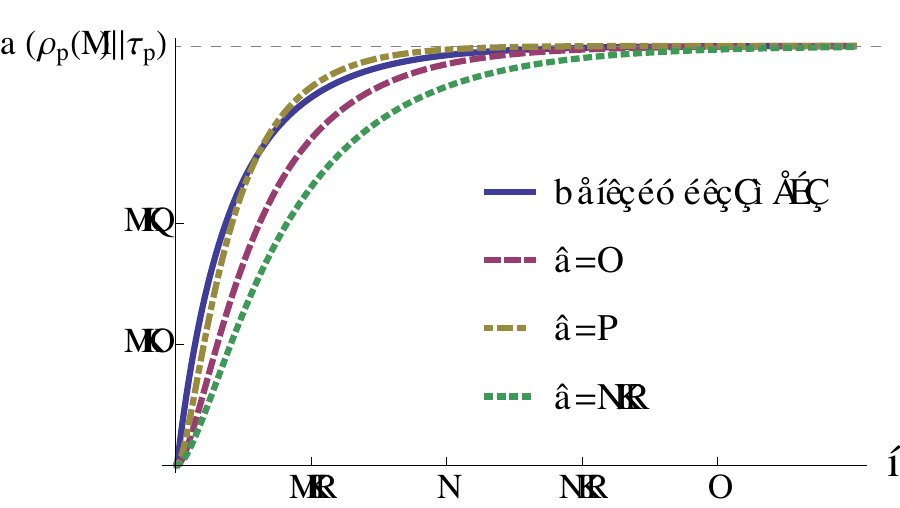}
\caption{Example plots for Theorem \ref{lemmak} for the Davies map for qubits from \cite{roga2010davies}. The solid (blue) curve is the amount of entropy produced $\beta F_\beta(\rho(0))-\beta F_\beta(\rho(t))$ (l.h.s. of Eq. \eqref{eq:largek}) while the dashed lines correspond to $D(\rho(0)\|\rho(k t))$ (r.h.s. of Eq. \eqref{eq:largek}) for different $k$. We see that when the constant $k$ is greater than $2$ the bound does not hold any more, showing that the $k=2$ case is indeed special. For $k<2$ the bound holds intuitively (given that it holds for $k=2$), but results in a worse bound. This shows that the constraint set by Eq. \eqref{Eq:proof theorem} reflects a special feature of how Davies maps thermalize. Moreover, we see that a $k>2$ would predict (incorrectly) a faster thermalisation rate, thus confirming that Eq. \eqref{Eq:proof theorem} is an implicit universal bound on the rate of thermalisation for Davies maps. The $y$ axis is dimensionless and the $x$ axis is in units of the inverse of the coupling constant of the semigroup. \label{fig:Plot2}}
\end{figure}

\section{Beyond Davies maps}
We now turn our attention to what recent developments from quantum information theory can say about convergence of dynamical semi-groups in general. 
A recent advancement in quantum information is the development of \textit{universal recoverability} maps \cite{wilde2015recoverability,junge2015universal,sutter2016multivariate}. By universal recoverability, it is meant that given a state $\sigma$ and a CPTP map $\Gamma$, one can use the recovery map to lower bound the relative entropy difference $D(\rho\|\sigma)-D(\Gamma(\rho)\| \Gamma(\sigma))$ for \textit{all} quantum states $\rho$. In general the lower bound takes on a complicated form (see Appendix 
\ref{sec: appendix Maps Beyond Davies}). 
However, for the case of dynamical semi-groups satisfying QDB and the following property, the bound is more explicit.

Let us assume that we have a one-parameter dynamical semi-group $M_t(\cdot)$ equipped with a fixed point $\Omega$ that satisfies a condition we call \textit{Time-translation symmetry w.r.t. fixed point} (TTSFP), 
\begin{equation}\label{eq:tts}
\mathcal{L}(\cdot) = \Omega^{i t} \mathcal{L} \left(\Omega^{-i t} (\cdot) \Omega^{i t}\right) \Omega^{-i t} \,\,\, \forall t \in \mathbb{R}.
\end{equation}
This condition is satisfied for example by dynamical semi-groups which arise naturally in the weak coupling limit or the low-density limit. Davies maps are one such example, but there are others \cite{Mwolf_Cubitt}.

The properties lead to the following result:
\begin{theorem}
Let the Quantum Dynamical Semigroup $M_t(\cdot)$ satisfy QDB and TTSFP. Then the following holds
\begin{equation}\label{eq:junge2 main text}
D(\rho(0) || \Omega) - D(\rho(t) || \Omega) \ge -2 \log{F\big(\rho, M_t (\rho(t))\big)},
\end{equation}
where $F(\rho,\sigma)=\tr[\sqrt{\sqrt{\sigma}\rho \sqrt{\sigma}}]$ is the quantum fidelity.
Moreover, if the generators are time-independent we may write $M_t(\rho(t))=\rho(2 t)$.
\end{theorem}

It is well known that $D(\rho\|\sigma)\geq -2 \log F(\rho\|\sigma)$ with equality \textit{only} for special instances. Therefore, for Davies maps, Eq. \eqref{eq:junge2 main text} is satisfied but with a weaker bound than Theorem \ref{thorem:Li-Winter conjecture for Davies maps}.

\section{Conclusion}

One of the main features in the study of dynamical thermalisation processes, such as Davies maps, is QDB. By using tools from quantum information theory, we show that the entropy produced after a time $t$, is lower bounded by how well one can \textit{recover} the initial $\rho_S(0)$ state from the state $\rho_S(t)$ via a recovery map. We then show that, due to QDB, the best way to perform the recovery is to time evolve \textit{forward} in time an amount $t$ to the state $\rho_S(2t)$. What's more, if one time evolves $\rho_S(t)$ for time $t'<t$, a worse bound is generated, while if one evolves for $t'>t$, the bound is not true for all Davies maps; thus showing that the connection between reversibility and recoverability suggested by QBD leads to tight dynamical bounds.

One of the important questions regarding Davies maps is how fast they converge to equilibrium. There have been several approaches to this question, mostly inspired by their classical analogues, which include the computation of the spectral gaps \cite{Alicki_Fannes_Horodecki,temme2010chi2,temme2013lower} or the logarithmic-Sobolev inequalities \cite{kastoryano2013quantum,temme2014hypercontractivity}. In particular we note that the latter take the form of upper bounds on distance measures between $\rho_S(t)$ and the thermal state. Likewise Eq. \eqref{Eq:proof theorem} can be re-arranged to give an upper bound in terms of the  relative-entropy to the Gibbs state, $D(\rho_S(t)\|\tau_S) \leq D(\rho_S(0)\|\tau_S)-D\left(\rho_S(0) \big{\|} \rho_S(2t)\right))$. 
It would be interesting to know if the bound of Eq. \eqref{Eq:proof theorem}, for \textit{primitive} Davies Maps, i.e. the dynamics converge to a unique fixed point, contains information about their asymptotic convergence. For instance, one could look at how fast is the inequality saturated in particular cases. We however leave this for future work.

 Another potential application of our work in open quantum systems, is to use a tightened monotonicity inequality to find when \emph{information backflow} occurs in non-Markovian dynamics  \cite{angelRev,Angel_PRL,laine2010measure}.
 
The condition of detailed balance is ubiquitous in thermalization processes, and in particular, current algorithms for simulating thermal states on a quantum computer, such as the quantum Metropolis algorithm \cite{temme2011quantum}, obey it, which makes it all the more interesting. As such, the useful connection we establish here between the Petz recovery map and QDB, is likely to have further implications for both thermodynamics and information theory.
\vspace{0.3CM}
\acknowledgments
The authors would like to thank David Sutter, Michael  Wolf, Nilanjana Datta, Mark Wilde and Toby Cubitt for helpful discussions. AA and MW acknowledge support from FQXi and EPSRC. This work was partially supported by the COST Action MP1209.

\bibliographystyle{unsrt} \bibliographystyle{apsrev4-1}
\bibliography{References}

\begin{thebibliography}{58}%
\makeatletter
\providecommand \@ifxundefined [1]{%
 \@ifx{#1\undefined}
}%
\providecommand \@ifnum [1]{%
 \ifnum #1\expandafter \@firstoftwo
 \else \expandafter \@secondoftwo
 \fi
}%
\providecommand \@ifx [1]{%
 \ifx #1\expandafter \@firstoftwo
 \else \expandafter \@secondoftwo
 \fi
}%
\providecommand \natexlab [1]{#1}%
\providecommand \enquote  [1]{``#1''}%
\providecommand \bibnamefont  [1]{#1}%
\providecommand \bibfnamefont [1]{#1}%
\providecommand \citenamefont [1]{#1}%
\providecommand \href@noop [0]{\@secondoftwo}%
\providecommand \href [0]{\begingroup \@sanitize@url \@href}%
\providecommand \@href[1]{\@@startlink{#1}\@@href}%
\providecommand \@@href[1]{\endgroup#1\@@endlink}%
\providecommand \@sanitize@url [0]{\catcode `\\12\catcode `\$12\catcode
  `\&12\catcode `\#12\catcode `\^12\catcode `\_12\catcode `\%12\relax}%
\providecommand \@@startlink[1]{}%
\providecommand \@@endlink[0]{}%
\providecommand \url  [0]{\begingroup\@sanitize@url \@url }%
\providecommand \@url [1]{\endgroup\@href {#1}{\urlprefix }}%
\providecommand \urlprefix  [0]{URL }%
\providecommand \Eprint [0]{\href }%
\providecommand \doibase [0]{http://dx.doi.org/}%
\providecommand \selectlanguage [0]{\@gobble}%
\providecommand \bibinfo  [0]{\@secondoftwo}%
\providecommand \bibfield  [0]{\@secondoftwo}%
\providecommand \translation [1]{[#1]}%
\providecommand \BibitemOpen [0]{}%
\providecommand \bibitemStop [0]{}%
\providecommand \bibitemNoStop [0]{.\EOS\space}%
\providecommand \EOS [0]{\spacefactor3000\relax}%
\providecommand \BibitemShut  [1]{\csname bibitem#1\endcsname}%
\let\auto@bib@innerbib\@empty
\bibitem [{\citenamefont {Esposito}\ \emph {et~al.}(2010)\citenamefont
  {Esposito}, \citenamefont {Lindenberg},\ and\ \citenamefont {Van~den
  Broeck}}]{esposito2010entropy}%
  \BibitemOpen
  \bibfield  {author} {\bibinfo {author} {\bibfnamefont {M.}~\bibnamefont
  {Esposito}}, \bibinfo {author} {\bibfnamefont {K.}~\bibnamefont
  {Lindenberg}}, \ and\ \bibinfo {author} {\bibfnamefont {C.}~\bibnamefont
  {Van~den Broeck}},\ }\href@noop {} {\bibfield  {journal} {\bibinfo  {journal}
  {New Journal of Physics}\ }\textbf {\bibinfo {volume} {12}},\ \bibinfo
  {pages} {013013} (\bibinfo {year} {2010})}\BibitemShut {NoStop}%
\bibitem [{\citenamefont {Davies}(1974)}]{davies1974}%
  \BibitemOpen
  \bibfield  {author} {\bibinfo {author} {\bibfnamefont {E.}~\bibnamefont
  {Davies}},\ }\href@noop {} {\bibfield  {journal} {\bibinfo  {journal}
  {Commun. Math. Phys.}\ }\textbf {\bibinfo {volume} {39}},\ \bibinfo {pages}
  {91 } (\bibinfo {year} {1974})}\BibitemShut {NoStop}%
\bibitem [{\citenamefont {Spohn}(1978)}]{spohn1978entropy}%
  \BibitemOpen
  \bibfield  {author} {\bibinfo {author} {\bibfnamefont {H.}~\bibnamefont
  {Spohn}},\ }\href@noop {} {\bibfield  {journal} {\bibinfo  {journal} {Journal
  of Mathematical Physics}\ }\textbf {\bibinfo {volume} {19}},\ \bibinfo
  {pages} {1227} (\bibinfo {year} {1978})}\BibitemShut {NoStop}%
\bibitem [{\citenamefont {Beigi}\ \emph {et~al.}(2016)\citenamefont {Beigi},
  \citenamefont {Datta},\ and\ \citenamefont {Leditzky}}]{beigi2015decoding}%
  \BibitemOpen
  \bibfield  {author} {\bibinfo {author} {\bibfnamefont {S.}~\bibnamefont
  {Beigi}}, \bibinfo {author} {\bibfnamefont {N.}~\bibnamefont {Datta}}, \ and\
  \bibinfo {author} {\bibfnamefont {F.}~\bibnamefont {Leditzky}},\ }\href
  {\doibase http://dx.doi.org/10.1063/1.4961515} {\bibfield  {journal}
  {\bibinfo  {journal} {Journal of Mathematical Physics}\ }\textbf {\bibinfo
  {volume} {57}},\ \bibinfo {eid} {082203} (\bibinfo {year} {2016}),\
  http://dx.doi.org/10.1063/1.4961515}\BibitemShut {NoStop}%
\bibitem [{\citenamefont {Brandao}\ and\ \citenamefont
  {Kastoryano}(2016)}]{Brandao2016}%
  \BibitemOpen
  \bibfield  {author} {\bibinfo {author} {\bibfnamefont {F.~G. S.~L.}\
  \bibnamefont {Brandao}}\ and\ \bibinfo {author} {\bibfnamefont {M.~J.}\
  \bibnamefont {Kastoryano}},\ }\href {http://arxiv.org/abs/1609.07877} {\
  (\bibinfo {year} {2016})},\ \Eprint {http://arxiv.org/abs/1609.07877}
  {arXiv:1609.07877} \BibitemShut {NoStop}%
\bibitem [{\citenamefont {Ng}\ and\ \citenamefont
  {Mandayam}(2010)}]{ng2010simple}%
  \BibitemOpen
  \bibfield  {author} {\bibinfo {author} {\bibfnamefont {H.~K.}\ \bibnamefont
  {Ng}}\ and\ \bibinfo {author} {\bibfnamefont {P.}~\bibnamefont {Mandayam}},\
  }\href@noop {} {\bibfield  {journal} {\bibinfo  {journal} {Physical Review
  A}\ }\textbf {\bibinfo {volume} {81}},\ \bibinfo {pages} {062342} (\bibinfo
  {year} {2010})}\BibitemShut {NoStop}%
\bibitem [{\citenamefont {Marvian}\ and\ \citenamefont
  {Lloyd}(2016)}]{marvian2016clocks}%
  \BibitemOpen
  \bibfield  {author} {\bibinfo {author} {\bibfnamefont {I.}~\bibnamefont
  {Marvian}}\ and\ \bibinfo {author} {\bibfnamefont {S.}~\bibnamefont
  {Lloyd}},\ }\href@noop {} {\bibfield  {journal} {\bibinfo  {journal} {arXiv
  preprint arXiv:1608.07325}\ } (\bibinfo {year} {2016})}\BibitemShut {NoStop}%
\bibitem [{\citenamefont {Aberg}(2016)}]{aberg2016fully}%
  \BibitemOpen
  \bibfield  {author} {\bibinfo {author} {\bibfnamefont {J.}~\bibnamefont
  {Aberg}},\ }\href@noop {} {\bibfield  {journal} {\bibinfo  {journal} {arXiv
  preprint arXiv:1601.01302}\ } (\bibinfo {year} {2016})}\BibitemShut {NoStop}%
\bibitem [{\citenamefont {Alhambra}\ \emph {et~al.}(2016)\citenamefont
  {Alhambra}, \citenamefont {Masanes}, \citenamefont {Oppenheim},\ and\
  \citenamefont {Perry}}]{alhambra2016fluctuating}%
  \BibitemOpen
  \bibfield  {author} {\bibinfo {author} {\bibfnamefont {{\'A}.~M.}\
  \bibnamefont {Alhambra}}, \bibinfo {author} {\bibfnamefont {L.}~\bibnamefont
  {Masanes}}, \bibinfo {author} {\bibfnamefont {J.}~\bibnamefont {Oppenheim}},
  \ and\ \bibinfo {author} {\bibfnamefont {C.}~\bibnamefont {Perry}},\
  }\href@noop {} {\bibfield  {journal} {\bibinfo  {journal} {Physical Review
  X}\ }\textbf {\bibinfo {volume} {6}},\ \bibinfo {pages} {041017} (\bibinfo
  {year} {2016})}\BibitemShut {NoStop}%
\bibitem [{\citenamefont {Landauer}(1961)}]{landauer1961irreversibility}%
  \BibitemOpen
  \bibfield  {author} {\bibinfo {author} {\bibfnamefont {R.}~\bibnamefont
  {Landauer}},\ }\href {\doibase 10.1147/rd.53.0183} {\bibfield  {journal}
  {\bibinfo  {journal} {IBM Journal of Research and Development}\ }\textbf
  {\bibinfo {volume} {5}},\ \bibinfo {pages} {183} (\bibinfo {year}
  {1961})}\BibitemShut {NoStop}%
\bibitem [{\citenamefont {Goold}\ \emph {et~al.}(2016)\citenamefont {Goold},
  \citenamefont {Huber}, \citenamefont {Riera}, \citenamefont {del Rio},\ and\
  \citenamefont {Skrzypczyk}}]{goold2016role}%
  \BibitemOpen
  \bibfield  {author} {\bibinfo {author} {\bibfnamefont {J.}~\bibnamefont
  {Goold}}, \bibinfo {author} {\bibfnamefont {M.}~\bibnamefont {Huber}},
  \bibinfo {author} {\bibfnamefont {A.}~\bibnamefont {Riera}}, \bibinfo
  {author} {\bibfnamefont {L.}~\bibnamefont {del Rio}}, \ and\ \bibinfo
  {author} {\bibfnamefont {P.}~\bibnamefont {Skrzypczyk}},\ }\href
  {http://stacks.iop.org/1751-8121/49/i=14/a=143001} {\bibfield  {journal}
  {\bibinfo  {journal} {Journal of Physics A: Mathematical and Theoretical}\
  }\textbf {\bibinfo {volume} {49}},\ \bibinfo {pages} {143001} (\bibinfo
  {year} {2016})}\BibitemShut {NoStop}%
\bibitem [{\citenamefont {Alicki}\ and\ \citenamefont
  {Lendi}(2007)}]{alicki2007quantum}%
  \BibitemOpen
  \bibfield  {author} {\bibinfo {author} {\bibfnamefont {R.}~\bibnamefont
  {Alicki}}\ and\ \bibinfo {author} {\bibfnamefont {K.}~\bibnamefont {Lendi}},\
  }\href@noop {} {\emph {\bibinfo {title} {Quantum Dynamical Semigroups and
  Applications}}},\ Vol.\ \bibinfo {volume} {717}\ (\bibinfo  {publisher}
  {Springer Science \& Business Media},\ \bibinfo {year} {2007})\BibitemShut
  {NoStop}%
\bibitem [{\citenamefont {Temme}(2013)}]{temme2013lower}%
  \BibitemOpen
  \bibfield  {author} {\bibinfo {author} {\bibfnamefont {K.}~\bibnamefont
  {Temme}},\ }\href@noop {} {\bibfield  {journal} {\bibinfo  {journal} {Journal
  of Mathematical Physics}\ }\textbf {\bibinfo {volume} {54}},\ \bibinfo
  {pages} {122110} (\bibinfo {year} {2013})}\BibitemShut {NoStop}%
\bibitem [{\citenamefont {Kossakowski}\ \emph {et~al.}(1977)\citenamefont
  {Kossakowski}, \citenamefont {Frigerio}, \citenamefont {Gorini},\ and\
  \citenamefont {Verri}}]{kossakowski1977}%
  \BibitemOpen
  \bibfield  {author} {\bibinfo {author} {\bibfnamefont {A.}~\bibnamefont
  {Kossakowski}}, \bibinfo {author} {\bibfnamefont {A.}~\bibnamefont
  {Frigerio}}, \bibinfo {author} {\bibfnamefont {V.}~\bibnamefont {Gorini}}, \
  and\ \bibinfo {author} {\bibfnamefont {M.}~\bibnamefont {Verri}},\ }\href
  {http://projecteuclid.org/euclid.cmp/1103901281} {\bibfield  {journal}
  {\bibinfo  {journal} {Comm. Math. Phys.}\ }\textbf {\bibinfo {volume} {57}},\
  \bibinfo {pages} {97} (\bibinfo {year} {1977})}\BibitemShut {NoStop}%
\bibitem [{\citenamefont {Spohn}\ and\ \citenamefont
  {Lebowitz}(2007)}]{SpohnHerbert07}%
  \BibitemOpen
  \bibfield  {author} {\bibinfo {author} {\bibfnamefont {H.}~\bibnamefont
  {Spohn}}\ and\ \bibinfo {author} {\bibfnamefont {J.~L.}\ \bibnamefont
  {Lebowitz}},\ }\enquote {\bibinfo {title} {Irreversible thermodynamics for
  quantum systems weakly coupled to thermal reservoirs},}\ in\ \href {\doibase
  10.1002/9780470142578.ch2} {\emph {\bibinfo {booktitle} {Advances in Chemical
  Physics}}}\ (\bibinfo  {publisher} {John Wiley \& Sons, Inc.},\ \bibinfo
  {year} {2007})\ pp.\ \bibinfo {pages} {109--142}\BibitemShut {NoStop}%
\bibitem [{\citenamefont {Spohn}(1977)}]{Spohn1977_converge1}%
  \BibitemOpen
  \bibfield  {author} {\bibinfo {author} {\bibfnamefont {H.}~\bibnamefont
  {Spohn}},\ }\href {\doibase 10.1007/BF00420668} {\bibfield  {journal}
  {\bibinfo  {journal} {Letters in Mathematical Physics}\ }\textbf {\bibinfo
  {volume} {2}},\ \bibinfo {pages} {33} (\bibinfo {year} {1977})}\BibitemShut
  {NoStop}%
\bibitem [{\citenamefont {Frigerio}(1977)}]{Frigerio1977_converge2}%
  \BibitemOpen
  \bibfield  {author} {\bibinfo {author} {\bibfnamefont {A.}~\bibnamefont
  {Frigerio}},\ }\href {\doibase 10.1007/BF00398571} {\bibfield  {journal}
  {\bibinfo  {journal} {Letters in Mathematical Physics}\ }\textbf {\bibinfo
  {volume} {2}},\ \bibinfo {pages} {79} (\bibinfo {year} {1977})}\BibitemShut
  {NoStop}%
\bibitem [{\citenamefont {Frigerio}(1978)}]{Frigerio1978_converge3}%
  \BibitemOpen
  \bibfield  {author} {\bibinfo {author} {\bibfnamefont {A.}~\bibnamefont
  {Frigerio}},\ }\href {\doibase 10.1007/BF01196936} {\bibfield  {journal}
  {\bibinfo  {journal} {Communications in Mathematical Physics}\ }\textbf
  {\bibinfo {volume} {63}},\ \bibinfo {pages} {269} (\bibinfo {year}
  {1978})}\BibitemShut {NoStop}%
\bibitem [{\citenamefont {Buscemi}\ \emph {et~al.}(2016)\citenamefont
  {Buscemi}, \citenamefont {Das},\ and\ \citenamefont
  {Wilde}}]{buscemi2016approximate}%
  \BibitemOpen
  \bibfield  {author} {\bibinfo {author} {\bibfnamefont {F.}~\bibnamefont
  {Buscemi}}, \bibinfo {author} {\bibfnamefont {S.}~\bibnamefont {Das}}, \ and\
  \bibinfo {author} {\bibfnamefont {M.~M.}\ \bibnamefont {Wilde}},\ }\href
  {\doibase 10.1103/PhysRevA.93.062314} {\bibfield  {journal} {\bibinfo
  {journal} {Phys. Rev. A}\ }\textbf {\bibinfo {volume} {93}},\ \bibinfo
  {pages} {062314} (\bibinfo {year} {2016})}\BibitemShut {NoStop}%
\bibitem [{\citenamefont {Li}\ and\ \citenamefont
  {Winter}(2014)}]{li2014squashed}%
  \BibitemOpen
  \bibfield  {author} {\bibinfo {author} {\bibfnamefont {K.}~\bibnamefont
  {Li}}\ and\ \bibinfo {author} {\bibfnamefont {A.}~\bibnamefont {Winter}},\
  }\href@noop {} {\bibfield  {journal} {\bibinfo  {journal} {arXiv preprint
  arXiv:1410.4184}\ } (\bibinfo {year} {2014})}\BibitemShut {NoStop}%
\bibitem [{\citenamefont {Brandao}\ \emph {et~al.}(2015)\citenamefont
  {Brandao}, \citenamefont {Harrow}, \citenamefont {Oppenheim},\ and\
  \citenamefont {Strelchuk}}]{brandao2015quantum}%
  \BibitemOpen
  \bibfield  {author} {\bibinfo {author} {\bibfnamefont {F.~G.}\ \bibnamefont
  {Brandao}}, \bibinfo {author} {\bibfnamefont {A.~W.}\ \bibnamefont {Harrow}},
  \bibinfo {author} {\bibfnamefont {J.}~\bibnamefont {Oppenheim}}, \ and\
  \bibinfo {author} {\bibfnamefont {S.}~\bibnamefont {Strelchuk}},\ }\href@noop
  {} {\bibfield  {journal} {\bibinfo  {journal} {Physical review letters}\
  }\textbf {\bibinfo {volume} {115}},\ \bibinfo {pages} {050501} (\bibinfo
  {year} {2015})}\BibitemShut {NoStop}%
\bibitem [{\citenamefont {Alhambra}\ \emph {et~al.}(2015)\citenamefont
  {Alhambra}, \citenamefont {Wehner}, \citenamefont {Wilde},\ and\
  \citenamefont {Woods}}]{wehner2015work}%
  \BibitemOpen
  \bibfield  {author} {\bibinfo {author} {\bibfnamefont {{\'A}.~M.}\
  \bibnamefont {Alhambra}}, \bibinfo {author} {\bibfnamefont {S.}~\bibnamefont
  {Wehner}}, \bibinfo {author} {\bibfnamefont {M.~M.}\ \bibnamefont {Wilde}}, \
  and\ \bibinfo {author} {\bibfnamefont {M.~P.}\ \bibnamefont {Woods}},\
  }\href@noop {} {\bibfield  {journal} {\bibinfo  {journal} {arXiv preprint
  arXiv:1506.08145}\ } (\bibinfo {year} {2015})}\BibitemShut {NoStop}%
\bibitem [{\citenamefont {Petz}(1986)}]{petz1986sufficient}%
  \BibitemOpen
  \bibfield  {author} {\bibinfo {author} {\bibfnamefont {D.}~\bibnamefont
  {Petz}},\ }\href@noop {} {\bibfield  {journal} {\bibinfo  {journal}
  {Communications in mathematical physics}\ }\textbf {\bibinfo {volume}
  {105}},\ \bibinfo {pages} {123} (\bibinfo {year} {1986})}\BibitemShut
  {NoStop}%
\bibitem [{\citenamefont {Wilde}(2015)}]{wilde2015recoverability}%
  \BibitemOpen
  \bibfield  {author} {\bibinfo {author} {\bibfnamefont {M.~M.}\ \bibnamefont
  {Wilde}},\ }in\ \href@noop {} {\emph {\bibinfo {booktitle} {Proc. R. Soc.
  A}}},\ Vol.\ \bibinfo {volume} {471}\ (\bibinfo {organization} {The Royal
  Society},\ \bibinfo {year} {2015})\ p.\ \bibinfo {pages}
  {20150338}\BibitemShut {NoStop}%
\bibitem [{\citenamefont {Junge}\ \emph {et~al.}(2015)\citenamefont {Junge},
  \citenamefont {Renner}, \citenamefont {Sutter}, \citenamefont {Wilde},\ and\
  \citenamefont {Winter}}]{junge2015universal}%
  \BibitemOpen
  \bibfield  {author} {\bibinfo {author} {\bibfnamefont {M.}~\bibnamefont
  {Junge}}, \bibinfo {author} {\bibfnamefont {R.}~\bibnamefont {Renner}},
  \bibinfo {author} {\bibfnamefont {D.}~\bibnamefont {Sutter}}, \bibinfo
  {author} {\bibfnamefont {M.~M.}\ \bibnamefont {Wilde}}, \ and\ \bibinfo
  {author} {\bibfnamefont {A.}~\bibnamefont {Winter}},\ }\href@noop {}
  {\bibfield  {journal} {\bibinfo  {journal} {arXiv preprint arXiv:1509.07127}\
  } (\bibinfo {year} {2015})}\BibitemShut {NoStop}%
\bibitem [{\citenamefont {Sutter}\ \emph
  {et~al.}(2016{\natexlab{a}})\citenamefont {Sutter}, \citenamefont
  {Tomamichel},\ and\ \citenamefont {Harrow}}]{sutter2016strengthened}%
  \BibitemOpen
  \bibfield  {author} {\bibinfo {author} {\bibfnamefont {D.}~\bibnamefont
  {Sutter}}, \bibinfo {author} {\bibfnamefont {M.}~\bibnamefont {Tomamichel}},
  \ and\ \bibinfo {author} {\bibfnamefont {A.~W.}\ \bibnamefont {Harrow}},\
  }\href@noop {} {\bibfield  {journal} {\bibinfo  {journal} {IEEE Transactions
  on Information Theory}\ }\textbf {\bibinfo {volume} {62}},\ \bibinfo {pages}
  {2907} (\bibinfo {year} {2016}{\natexlab{a}})}\BibitemShut {NoStop}%
\bibitem [{\citenamefont {Sutter}\ \emph
  {et~al.}(2016{\natexlab{b}})\citenamefont {Sutter}, \citenamefont {Berta},\
  and\ \citenamefont {Tomamichel}}]{sutter2016multivariate}%
  \BibitemOpen
  \bibfield  {author} {\bibinfo {author} {\bibfnamefont {D.}~\bibnamefont
  {Sutter}}, \bibinfo {author} {\bibfnamefont {M.}~\bibnamefont {Berta}}, \
  and\ \bibinfo {author} {\bibfnamefont {M.}~\bibnamefont {Tomamichel}},\
  }\href@noop {} {\bibfield  {journal} {\bibinfo  {journal} {Communications in
  Mathematical Physics}\ ,\ \bibinfo {pages} {1}} (\bibinfo {year}
  {2016}{\natexlab{b}})}\BibitemShut {NoStop}%
\bibitem [{\citenamefont {Boltzmann}(1964)}]{Boltzmann1964}%
  \BibitemOpen
  \bibfield  {author} {\bibinfo {author} {\bibfnamefont {L.}~\bibnamefont
  {Boltzmann}},\ }\href@noop {} {\emph {\bibinfo {title} {{Lectures on gas
  theory; Berkeley}}}}\ (\bibinfo  {publisher} {University of California Press,
  Berkeley, CA},\ \bibinfo {year} {1964})\BibitemShut {NoStop}%
\bibitem [{\citenamefont {Alicki}(1976)}]{Alicki_DB_first}%
  \BibitemOpen
  \bibfield  {author} {\bibinfo {author} {\bibfnamefont {R.}~\bibnamefont
  {Alicki}},\ }\href {\doibase http://dx.doi.org/10.1016/0034-4877(76)90046-X}
  {\bibfield  {journal} {\bibinfo  {journal} {Reports on Mathematical Physics}\
  }\textbf {\bibinfo {volume} {10}},\ \bibinfo {pages} {249 } (\bibinfo {year}
  {1976})}\BibitemShut {NoStop}%
\bibitem [{\citenamefont {Fagnola}\ and\ \citenamefont
  {Umanit\`a}(2007)}]{oftherDBdef}%
  \BibitemOpen
  \bibfield  {author} {\bibinfo {author} {\bibfnamefont {F.}~\bibnamefont
  {Fagnola}}\ and\ \bibinfo {author} {\bibfnamefont {V.}~\bibnamefont
  {Umanit\`a}},\ }\href {\doibase 10.1142/S0219025707002762} {\bibfield
  {journal} {\bibinfo  {journal} {Infinite Dimensional Analysis, Quantum
  Probability and Related Topics}\ }\textbf {\bibinfo {volume} {10}},\ \bibinfo
  {pages} {335} (\bibinfo {year} {2007})}\BibitemShut {NoStop}%
\bibitem [{\citenamefont {Temme}\ \emph {et~al.}(2010)\citenamefont {Temme},
  \citenamefont {Kastoryano}, \citenamefont {Ruskai}, \citenamefont {Wolf},\
  and\ \citenamefont {Verstraete}}]{temme2010chi2}%
  \BibitemOpen
  \bibfield  {author} {\bibinfo {author} {\bibfnamefont {K.}~\bibnamefont
  {Temme}}, \bibinfo {author} {\bibfnamefont {M.~J.}\ \bibnamefont
  {Kastoryano}}, \bibinfo {author} {\bibfnamefont {M.}~\bibnamefont {Ruskai}},
  \bibinfo {author} {\bibfnamefont {M.~M.}\ \bibnamefont {Wolf}}, \ and\
  \bibinfo {author} {\bibfnamefont {F.}~\bibnamefont {Verstraete}},\
  }\href@noop {} {\bibfield  {journal} {\bibinfo  {journal} {Journal of
  Mathematical Physics}\ }\textbf {\bibinfo {volume} {51}},\ \bibinfo {pages}
  {122201} (\bibinfo {year} {2010})}\BibitemShut {NoStop}%
\bibitem [{\citenamefont {Crooks}(2008)}]{crooks2008quantum}%
  \BibitemOpen
  \bibfield  {author} {\bibinfo {author} {\bibfnamefont {G.~E.}\ \bibnamefont
  {Crooks}},\ }\href@noop {} {\bibfield  {journal} {\bibinfo  {journal}
  {Physical Review A}\ }\textbf {\bibinfo {volume} {77}},\ \bibinfo {pages}
  {034101} (\bibinfo {year} {2008})}\BibitemShut {NoStop}%
\bibitem [{\citenamefont {Ticozzi}\ and\ \citenamefont
  {Pavon}(2010)}]{Ticozzi2010}%
  \BibitemOpen
  \bibfield  {author} {\bibinfo {author} {\bibfnamefont {F.}~\bibnamefont
  {Ticozzi}}\ and\ \bibinfo {author} {\bibfnamefont {M.}~\bibnamefont
  {Pavon}},\ }\href {\doibase 10.1007/s11128-010-0186-x} {\bibfield  {journal}
  {\bibinfo  {journal} {Quantum Information Processing}\ }\textbf {\bibinfo
  {volume} {9}},\ \bibinfo {pages} {551} (\bibinfo {year} {2010})}\BibitemShut
  {NoStop}%
\bibitem [{\citenamefont {Petz}(1988)}]{petz1988sufficiency}%
  \BibitemOpen
  \bibfield  {author} {\bibinfo {author} {\bibfnamefont {D.}~\bibnamefont
  {Petz}},\ }\href@noop {} {\bibfield  {journal} {\bibinfo  {journal} {The
  Quarterly Journal of Mathematics}\ }\textbf {\bibinfo {volume} {39}},\
  \bibinfo {pages} {97} (\bibinfo {year} {1988})}\BibitemShut {NoStop}%
\bibitem [{\citenamefont {Petz}(2003)}]{petz2003monotonicity}%
  \BibitemOpen
  \bibfield  {author} {\bibinfo {author} {\bibfnamefont {D.}~\bibnamefont
  {Petz}},\ }\href@noop {} {\bibfield  {journal} {\bibinfo  {journal} {Reviews
  in Mathematical Physics}\ }\textbf {\bibinfo {volume} {15}},\ \bibinfo
  {pages} {79} (\bibinfo {year} {2003})}\BibitemShut {NoStop}%
\bibitem [{\citenamefont {Barnum}\ and\ \citenamefont
  {Knill}(2002)}]{barnum2002reversing}%
  \BibitemOpen
  \bibfield  {author} {\bibinfo {author} {\bibfnamefont {H.}~\bibnamefont
  {Barnum}}\ and\ \bibinfo {author} {\bibfnamefont {E.}~\bibnamefont {Knill}},\
  }\href@noop {} {\bibfield  {journal} {\bibinfo  {journal} {Journal of
  Mathematical Physics}\ }\textbf {\bibinfo {volume} {43}},\ \bibinfo {pages}
  {2097} (\bibinfo {year} {2002})}\BibitemShut {NoStop}%
\bibitem [{\citenamefont {Wilde}(2013)}]{wilde2013quantum}%
  \BibitemOpen
  \bibfield  {author} {\bibinfo {author} {\bibfnamefont {M.~M.}\ \bibnamefont
  {Wilde}},\ }\href@noop {} {\emph {\bibinfo {title} {Quantum information
  theory}}}\ (\bibinfo  {publisher} {Cambridge University Press},\ \bibinfo
  {year} {2013})\BibitemShut {NoStop}%
\bibitem [{\citenamefont {Roga}\ \emph {et~al.}(2010)\citenamefont {Roga},
  \citenamefont {Fannes},\ and\ \citenamefont
  {{\.Z}yczkowski}}]{roga2010davies}%
  \BibitemOpen
  \bibfield  {author} {\bibinfo {author} {\bibfnamefont {W.}~\bibnamefont
  {Roga}}, \bibinfo {author} {\bibfnamefont {M.}~\bibnamefont {Fannes}}, \ and\
  \bibinfo {author} {\bibfnamefont {K.}~\bibnamefont {{\.Z}yczkowski}},\
  }\href@noop {} {\bibfield  {journal} {\bibinfo  {journal} {Reports on
  Mathematical Physics}\ }\textbf {\bibinfo {volume} {66}},\ \bibinfo {pages}
  {311} (\bibinfo {year} {2010})}\BibitemShut {NoStop}%
\bibitem [{\citenamefont {Kastoryano}\ and\ \citenamefont
  {Temme}(2013)}]{kastoryano2013quantum}%
  \BibitemOpen
  \bibfield  {author} {\bibinfo {author} {\bibfnamefont {M.~J.}\ \bibnamefont
  {Kastoryano}}\ and\ \bibinfo {author} {\bibfnamefont {K.}~\bibnamefont
  {Temme}},\ }\href@noop {} {\bibfield  {journal} {\bibinfo  {journal} {Journal
  of Mathematical Physics}\ }\textbf {\bibinfo {volume} {54}},\ \bibinfo
  {pages} {052202} (\bibinfo {year} {2013})}\BibitemShut {NoStop}%
\bibitem [{\citenamefont {Wolf}\ \emph {et~al.}(2008)\citenamefont {Wolf},
  \citenamefont {Eisert}, \citenamefont {Cubitt},\ and\ \citenamefont
  {Cirac}}]{Mwolf_Cubitt}%
  \BibitemOpen
  \bibfield  {author} {\bibinfo {author} {\bibfnamefont {M.~M.}\ \bibnamefont
  {Wolf}}, \bibinfo {author} {\bibfnamefont {J.}~\bibnamefont {Eisert}},
  \bibinfo {author} {\bibfnamefont {T.~S.}\ \bibnamefont {Cubitt}}, \ and\
  \bibinfo {author} {\bibfnamefont {J.~I.}\ \bibnamefont {Cirac}},\ }\href
  {\doibase 10.1103/PhysRevLett.101.150402} {\bibfield  {journal} {\bibinfo
  {journal} {Phys. Rev. Lett.}\ }\textbf {\bibinfo {volume} {101}},\ \bibinfo
  {pages} {150402} (\bibinfo {year} {2008})}\BibitemShut {NoStop}%
\bibitem [{\citenamefont {Alicki}\ \emph {et~al.}(2009)\citenamefont {Alicki},
  \citenamefont {Fannes},\ and\ \citenamefont
  {Horodecki}}]{Alicki_Fannes_Horodecki}%
  \BibitemOpen
  \bibfield  {author} {\bibinfo {author} {\bibfnamefont {R.}~\bibnamefont
  {Alicki}}, \bibinfo {author} {\bibfnamefont {M.}~\bibnamefont {Fannes}}, \
  and\ \bibinfo {author} {\bibfnamefont {M.}~\bibnamefont {Horodecki}},\ }\href
  {http://stacks.iop.org/1751-8121/42/i=6/a=065303} {\bibfield  {journal}
  {\bibinfo  {journal} {Journal of Physics A: Mathematical and Theoretical}\
  }\textbf {\bibinfo {volume} {42}},\ \bibinfo {pages} {065303} (\bibinfo
  {year} {2009})}\BibitemShut {NoStop}%
\bibitem [{\citenamefont {Temme}\ \emph {et~al.}(2014)\citenamefont {Temme},
  \citenamefont {Pastawski},\ and\ \citenamefont
  {Kastoryano}}]{temme2014hypercontractivity}%
  \BibitemOpen
  \bibfield  {author} {\bibinfo {author} {\bibfnamefont {K.}~\bibnamefont
  {Temme}}, \bibinfo {author} {\bibfnamefont {F.}~\bibnamefont {Pastawski}}, \
  and\ \bibinfo {author} {\bibfnamefont {M.~J.}\ \bibnamefont {Kastoryano}},\
  }\href@noop {} {\bibfield  {journal} {\bibinfo  {journal} {Journal of Physics
  A: Mathematical and Theoretical}\ }\textbf {\bibinfo {volume} {47}},\
  \bibinfo {pages} {405303} (\bibinfo {year} {2014})}\BibitemShut {NoStop}%
\bibitem [{\citenamefont {Rivas}\ \emph {et~al.}(2014)\citenamefont {Rivas},
  \citenamefont {Huelga},\ and\ \citenamefont {Plenio}}]{angelRev}%
  \BibitemOpen
  \bibfield  {author} {\bibinfo {author} {\bibfnamefont {{\'A}.}~\bibnamefont
  {Rivas}}, \bibinfo {author} {\bibfnamefont {S.~F.}\ \bibnamefont {Huelga}}, \
  and\ \bibinfo {author} {\bibfnamefont {M.~B.}\ \bibnamefont {Plenio}},\
  }\href {http://stacks.iop.org/0034-4885/77/i=9/a=094001} {\bibfield
  {journal} {\bibinfo  {journal} {Reports on Progress in Physics}\ }\textbf
  {\bibinfo {volume} {77}},\ \bibinfo {pages} {094001} (\bibinfo {year}
  {2014})}\BibitemShut {NoStop}%
\bibitem [{\citenamefont {Rivas}\ \emph {et~al.}(2010)\citenamefont {Rivas},
  \citenamefont {Huelga},\ and\ \citenamefont {Plenio}}]{Angel_PRL}%
  \BibitemOpen
  \bibfield  {author} {\bibinfo {author} {\bibfnamefont {{\'A}.}~\bibnamefont
  {Rivas}}, \bibinfo {author} {\bibfnamefont {S.~F.}\ \bibnamefont {Huelga}}, \
  and\ \bibinfo {author} {\bibfnamefont {M.~B.}\ \bibnamefont {Plenio}},\
  }\href {\doibase 10.1103/PhysRevLett.105.050403} {\bibfield  {journal}
  {\bibinfo  {journal} {Phys. Rev. Lett.}\ }\textbf {\bibinfo {volume} {105}},\
  \bibinfo {pages} {050403} (\bibinfo {year} {2010})}\BibitemShut {NoStop}%
\bibitem [{\citenamefont {Laine}\ \emph {et~al.}(2010)\citenamefont {Laine},
  \citenamefont {Piilo},\ and\ \citenamefont {Breuer}}]{laine2010measure}%
  \BibitemOpen
  \bibfield  {author} {\bibinfo {author} {\bibfnamefont {E.-M.}\ \bibnamefont
  {Laine}}, \bibinfo {author} {\bibfnamefont {J.}~\bibnamefont {Piilo}}, \ and\
  \bibinfo {author} {\bibfnamefont {H.-P.}\ \bibnamefont {Breuer}},\
  }\href@noop {} {\bibfield  {journal} {\bibinfo  {journal} {Physical Review
  A}\ }\textbf {\bibinfo {volume} {81}},\ \bibinfo {pages} {062115} (\bibinfo
  {year} {2010})}\BibitemShut {NoStop}%
\bibitem [{\citenamefont {Temme}\ \emph {et~al.}(2011)\citenamefont {Temme},
  \citenamefont {Osborne}, \citenamefont {Vollbrecht}, \citenamefont {Poulin},\
  and\ \citenamefont {Verstraete}}]{temme2011quantum}%
  \BibitemOpen
  \bibfield  {author} {\bibinfo {author} {\bibfnamefont {K.}~\bibnamefont
  {Temme}}, \bibinfo {author} {\bibfnamefont {T.}~\bibnamefont {Osborne}},
  \bibinfo {author} {\bibfnamefont {K.~G.}\ \bibnamefont {Vollbrecht}},
  \bibinfo {author} {\bibfnamefont {D.}~\bibnamefont {Poulin}}, \ and\ \bibinfo
  {author} {\bibfnamefont {F.}~\bibnamefont {Verstraete}},\ }\href@noop {}
  {\bibfield  {journal} {\bibinfo  {journal} {Nature}\ }\textbf {\bibinfo
  {volume} {471}},\ \bibinfo {pages} {87} (\bibinfo {year} {2011})}\BibitemShut
  {NoStop}%
\bibitem [{\citenamefont {Davies}(1979)}]{davies1979generators}%
  \BibitemOpen
  \bibfield  {author} {\bibinfo {author} {\bibfnamefont {E.~B.}\ \bibnamefont
  {Davies}},\ }\href@noop {} {\bibfield  {journal} {\bibinfo  {journal}
  {Journal of Functional Analysis}\ }\textbf {\bibinfo {volume} {34}},\
  \bibinfo {pages} {421} (\bibinfo {year} {1979})}\BibitemShut {NoStop}%
\bibitem [{\citenamefont {{\'C}wikli{\'n}ski}\ \emph
  {et~al.}(2015)\citenamefont {{\'C}wikli{\'n}ski}, \citenamefont
  {Studzi{\'n}ski}, \citenamefont {Horodecki},\ and\ \citenamefont
  {Oppenheim}}]{cwiklinski2015limitations}%
  \BibitemOpen
  \bibfield  {author} {\bibinfo {author} {\bibfnamefont {P.}~\bibnamefont
  {{\'C}wikli{\'n}ski}}, \bibinfo {author} {\bibfnamefont {M.}~\bibnamefont
  {Studzi{\'n}ski}}, \bibinfo {author} {\bibfnamefont {M.}~\bibnamefont
  {Horodecki}}, \ and\ \bibinfo {author} {\bibfnamefont {J.}~\bibnamefont
  {Oppenheim}},\ }\href@noop {} {\bibfield  {journal} {\bibinfo  {journal}
  {Physical review letters}\ }\textbf {\bibinfo {volume} {115}},\ \bibinfo
  {pages} {210403} (\bibinfo {year} {2015})}\BibitemShut {NoStop}%
\bibitem [{\citenamefont {Brand\~ao}\ \emph {et~al.}(2013)\citenamefont
  {Brand\~ao}, \citenamefont {Horodecki}, \citenamefont {Oppenheim},
  \citenamefont {Renes},\ and\ \citenamefont {Spekkens}}]{brandao2013resource}%
  \BibitemOpen
  \bibfield  {author} {\bibinfo {author} {\bibfnamefont {F.~G. S.~L.}\
  \bibnamefont {Brand\~ao}}, \bibinfo {author} {\bibfnamefont {M.}~\bibnamefont
  {Horodecki}}, \bibinfo {author} {\bibfnamefont {J.}~\bibnamefont
  {Oppenheim}}, \bibinfo {author} {\bibfnamefont {J.~M.}\ \bibnamefont
  {Renes}}, \ and\ \bibinfo {author} {\bibfnamefont {R.~W.}\ \bibnamefont
  {Spekkens}},\ }\href {\doibase 10.1103/PhysRevLett.111.250404} {\bibfield
  {journal} {\bibinfo  {journal} {Phys. Rev. Lett.}\ }\textbf {\bibinfo
  {volume} {111}},\ \bibinfo {pages} {250404} (\bibinfo {year}
  {2013})}\BibitemShut {NoStop}%
\bibitem [{\citenamefont {Rellich}(1953)}]{rellich}%
  \BibitemOpen
  \bibfield  {author} {\bibinfo {author} {\bibfnamefont {F.}~\bibnamefont
  {Rellich}},\ }\href@noop {} {\bibfield  {journal} {\bibinfo  {journal}
  {Lecture Notes reprinted by Gordon and Breach, 1968}\ } (\bibinfo {year}
  {1953})}\BibitemShut {NoStop}%
\bibitem [{\citenamefont {Kato}(1976)}]{kato}%
  \BibitemOpen
  \bibfield  {author} {\bibinfo {author} {\bibfnamefont {T.}~\bibnamefont
  {Kato}},\ }\href@noop {} {\bibfield  {journal} {\bibinfo  {journal}
  {SpringerVerlag, Berlin, New York}\ }\textbf {\bibinfo {volume} {132}}
  (\bibinfo {year} {1976})}\BibitemShut {NoStop}%
\bibitem [{\citenamefont {Ohya}\ and\ \citenamefont
  {Petz}(2004)}]{ohya2004quantum}%
  \BibitemOpen
  \bibfield  {author} {\bibinfo {author} {\bibfnamefont {M.}~\bibnamefont
  {Ohya}}\ and\ \bibinfo {author} {\bibfnamefont {D.}~\bibnamefont {Petz}},\
  }\href {https://books.google.co.uk/books?id=r2ullNVyESQC} {\emph {\bibinfo
  {title} {Quantum Entropy and Its Use}}},\ Theoretical and Mathematical
  Physics\ (\bibinfo  {publisher} {Springer Berlin Heidelberg},\ \bibinfo
  {year} {2004})\BibitemShut {NoStop}%
\bibitem [{\citenamefont {Breuer}\ and\ \citenamefont
  {Petruccione}(2002)}]{breuer2002theory}%
  \BibitemOpen
  \bibfield  {author} {\bibinfo {author} {\bibfnamefont {H.-P.}\ \bibnamefont
  {Breuer}}\ and\ \bibinfo {author} {\bibfnamefont {F.}~\bibnamefont
  {Petruccione}},\ }\href@noop {} {\emph {\bibinfo {title} {The theory of open
  quantum systems}}}\ (\bibinfo  {publisher} {Oxford University Press on
  Demand},\ \bibinfo {year} {2002})\BibitemShut {NoStop}%
\bibitem [{\citenamefont {Carlen}(2010)}]{carlen2010trace}%
  \BibitemOpen
  \bibfield  {author} {\bibinfo {author} {\bibfnamefont {E.}~\bibnamefont
  {Carlen}},\ }\href@noop {} {\emph {\bibinfo {title} {Trace inequalities and
  quantum entropy: an introductory course}}}\ (\bibinfo {year}
  {2010})\BibitemShut {NoStop}%
\bibitem [{\citenamefont {Kastoryano}\ and\ \citenamefont
  {Brandao}()}]{kastoryano2014quantum}%
  \BibitemOpen
  \bibfield  {author} {\bibinfo {author} {\bibfnamefont {M.~J.}\ \bibnamefont
  {Kastoryano}}\ and\ \bibinfo {author} {\bibfnamefont {F.~G.}\ \bibnamefont
  {Brandao}},\ }\href@noop {} {\bibinfo  {journal} {Communications in
  Mathematical Physics}\ ,\ \bibinfo {pages} {1}}\BibitemShut {NoStop}%
\bibitem [{\citenamefont {Derezinski}\ and\ \citenamefont
  {Fruboes}(2006)}]{G_rule}%
  \BibitemOpen
\bibfield  {journal} {  }\bibfield  {author} {\bibinfo {author} {\bibfnamefont
  {J.}~\bibnamefont {Derezinski}}\ and\ \bibinfo {author} {\bibfnamefont
  {F.}~\bibnamefont {Fruboes}},\ }\href@noop {} {\emph {\bibinfo {title}
  {Lecture Notes in Mathematics 1882 eds S. Attal, A. Joye, C.-A. Pillet}}}\
  (\bibinfo {year} {2006})\ pp.\ \bibinfo {pages} {67--116}\BibitemShut
  {NoStop}%
\bibitem [{\citenamefont {Marvian}\ and\ \citenamefont
  {Spekkens}(2014)}]{marvian2014modes}%
  \BibitemOpen
  \bibfield  {author} {\bibinfo {author} {\bibfnamefont {I.}~\bibnamefont
  {Marvian}}\ and\ \bibinfo {author} {\bibfnamefont {R.~W.}\ \bibnamefont
  {Spekkens}},\ }\href@noop {} {\bibfield  {journal} {\bibinfo  {journal}
  {Physical Review A}\ }\textbf {\bibinfo {volume} {90}},\ \bibinfo {pages}
  {062110} (\bibinfo {year} {2014})}\BibitemShut {NoStop}%
\bibitem [{\citenamefont {Lostaglio}\ \emph {et~al.}(2015)\citenamefont
  {Lostaglio}, \citenamefont {Korzekwa}, \citenamefont {Jennings},\ and\
  \citenamefont {Rudolph}}]{lostaglio2015quantum}%
  \BibitemOpen
  \bibfield  {author} {\bibinfo {author} {\bibfnamefont {M.}~\bibnamefont
  {Lostaglio}}, \bibinfo {author} {\bibfnamefont {K.}~\bibnamefont {Korzekwa}},
  \bibinfo {author} {\bibfnamefont {D.}~\bibnamefont {Jennings}}, \ and\
  \bibinfo {author} {\bibfnamefont {T.}~\bibnamefont {Rudolph}},\ }\href@noop
  {} {\bibfield  {journal} {\bibinfo  {journal} {Physical Review X}\ }\textbf
  {\bibinfo {volume} {5}},\ \bibinfo {pages} {021001} (\bibinfo {year}
  {2015})}\BibitemShut {NoStop}%
\end{thebibliography}%

\clearpage
\widetext
\appendix
\section{Technical results}\label{App1}
\subsection{Davies maps and conditions for Lemma \ref{thorem:Li-Winter conjecture for Davies maps1}}\label{sec:Conditions for Theorem Li-winter conjec}

Davies maps are derived from considering the dynamics of a state $\rho_S\in \mathcal{S}(\mathcal{H_S}),$ where $\mathcal{H_S}$ is of finite dimension $d_S$, in contact with a thermal bath on an infinite dimensional Hilbert space $\mathcal{H_B}$. We will here specify the minimal assumptions about the bath and its interaction with the system necessary for the derivation of Lemma \ref{lem:fixed pt almost comm} and Lemma \ref{thorem:Li-Winter conjecture for Davies maps1}. In order to guarantee other properties, such as the existence of a fixed point or detailed balance, more subtle constraints are also necessary.

Let $\hat H_B$ be a self-adjoint Hamiltonian on $\mathcal{H_B}$. Since we want states on $\hat H_B$ to be thermodynamically stable, we assume that $Z_B=\tr[\exp(-\beta \hat H_B)]<\infty$ for all $\beta>0$. $\hat H_B$ must therefore have a purely discrete spectrum, which is bounded below and has no finite limit points; that is, there are only a finite number of energy levels in any finite interval $\Delta E$. The  quantum state $\rho_S\in \mathcal{S}(\mathcal{H_S})$ with its free self-adjoint Hamiltonian $\hat H_S$ of finite dimension interacts with the system via a bounded interaction term $\hat I\in\mathcal{B}(\mathcal{H}_S\otimes\mathcal{H}_B)$, with a parameter $\lambda>0$ determining the interaction strength as follows
\be\label{eq:Hamil on S B} 
\hat H_{SB}=\hat H_S\otimes\id_B+\id_S\otimes\hat H_B+\lambda\,\hat I.
\ee
The initial state on $\mathcal{S}(\mathcal{H}_S\otimes\mathcal{H}_B)$ is assumed to be product, $\rho_S\otimes\tau_B$, with $\tau_B$ the Gibbs state at inverse temperature $\beta$. The dynamics of the system at time $\tilde t$ is given by the unitary operator 
\be
U(\tilde t):=\me^{-\mi \tilde t \hat H_{SB}}
\ee
after tracing out the environment. More precisely, by
\be 
\tr_B\left[ U(\tilde t) \rho_S\otimes\tau_B\, U^\dag(\tilde t) \right]\in\mathcal{S}(\mathcal{H}_S),
\ee
where $U^\dag$ denotes the adjoint of $U$.

The Davies map $T_t(\cdot)$ is defined by taking the limit that the interaction strength $\lambda$ goes to zero, while the time $\tilde t$ goes to infinity while maintaining $\tilde t \lambda^2:=t$ fixed. More concisely,
\be\label{eq:Davies map limit}
T_t(\cdot)=\lim_{\lambda\rightarrow 0^+} \tr_B\left[ U(\tilde t) (\cdot)\otimes\tau_B\, U^\dag(\tilde t)  \right]\in\mathcal{S}(\mathcal{H}_S) \quad \text{subject to } \tilde t \lambda^2=t  \text{ fixed.}
\ee
It is assumed that in this limit $U(\tilde t)$ and its inverse $U^\dag(\tilde t)$ are still unitary operators mapping states on $\mathcal{S}(\mathcal{H}_S\otimes\mathcal{H}_B)$ to states on $\mathcal{S}(\mathcal{H}_S\otimes\mathcal{H}_B)$. To gain more physical insight into this construction, we refer to \cite{SpohnHerbert07,davies1979generators,davies1974}. We remind the reader that the conditions described in Section \ref{sec:Conditions for Theorem Li-winter conjec} are not sufficient for the map $T_t(\cdot)$ to satisfy other properties, such as the convergence to a fixed point or detailed balance, more subtle constraints are also necessary. We will not go into the details of these additional conditions, since only sufficient (but perhaps not necessary) conditions are known, e.g. \cite{davies1974}. In other Sections, we will additionally take advantage of the known fact that Davies maps satisfy quantum detailed balance.
\subsection{Proof and statement of Lemma \ref{thorem:Li-Winter conjecture for Davies maps1}}\label{sec:main theorem}

In order to prove the main theorem we need a lemma about Davies maps first. 
We show that in the weak coupling limit, correlations between the system and the environment (the bath) are not created if both start as initially uncorrelated thermal states. In order to do this, we will need to introduce a finite dimensional cut-off on $\mathcal{H}_B$ and prove the results for the truncated space, and finally proving uniform convergence in the bath system size by removing the cut-off by taking the infinite dimensional limit.
Let $\hat P_n$ denote the projection onto a finite dimensional Hilbert Space $\mathcal{H}_{B,n}\subset \mathcal{H}_{B}$. Furthermore, assume that $\mathcal{H}_{B,1}\subset \mathcal{H}_{B,2}\subset \mathcal{H}_{B,3}\ldots$ and that $\lim_{n\rightarrow\infty}\mathcal{H}_{B,n}=\mathcal{H}_{B}$. For concreteness (although not strictly necessary), one could let $\hat P_n=\sum_{k=0}^n \ketbra{E_k}{E_k}$ where $\ket{E_0},\ket{E_1},\ket{E_2},\ldots$ are the eigenvectors of $\hat H_B$ ordered in increasing eigenvalue order.\\
We define the truncated self-adjoint Hamiltonians on $\mathcal{H_B}$ as $\hat H_B^{(n)}=\hat P_n \hat H_B \hat P_n$ with a corresponding Gibbs state denoted by $\tau_{B,n}\in\mathcal{S}(\mathcal{H}_{B,n})$. Similarly, we  construct unitaries on $\mathcal{H}_{B,n}$ by
\be\label{eq: unitary Ham SB}
U_n=\exp\left(-\mi \Delta \hat H_{SB}^{(n)}\right), \quad \hat H_{SB}^{(n)}= (\id_S\otimes\hat P_n)\, \hat H_{SB}\,(\id_S\otimes \hat P_n)
\ee
and define $\hat I_n:=(\id_S\otimes\hat P_n)\, \hat I\,(\id_S\otimes \hat P_n)$. We recall the definition of the thermal state of the system $\tau_S\in\mathcal{S}(\mathcal{H}_S)$, which is given by 
\be 
\tau_S=\frac{\me^{\beta \hat H_S}}{Z_S},\quad Z_S>0
\ee
for some inverse temperature $\beta>0$

The lemma is the following:
\begin{lemma}[Correlations at the fixed point]\label{lem:fixed pt almost comm}
	Let $\alpha> 0$, $\Delta\in \rr$ and the constant $\tilde Z_{SB}^{n,\alpha}=\tr[(\tau_S\otimes\tau_{B,n})^\alpha]$. Then, for all $n\in\nn^+$, we have the bound
	\be 
	\frac{1}{2}\|U_n(\tau_S\otimes\tau_{B,n})^\alpha U_n^\dag-(\tau_S\otimes\tau_{B,n})^\alpha\|_1\leq \tilde Z_{SB}^{n,\alpha}\beta\sqrt{\lambda \|\hat I_n\|},
	\ee 
	where $\tau_S$, $\tau_{B,n}$ are thermal states at inverse temperatures $\beta_S$, $\beta_{B,n}$ respectively, and $\|\cdot\|_1$, $\| \cdot\|$ is the one-norm and operator norm respectively.
\end{lemma}
\begin{proof}
	The result is a consequence of mean energy conservation under the unitary transformation $U_n$ and Pinsker's inequality.\\
	Define the shorthand notation $\tilde\tau^{n,\alpha}_{SB}=U_n(\tau_S\otimes\tau_{B,n})^{\alpha} U_n^\dag/\tilde Z_{SB}^{n,\alpha} \in\mathcal{S}(\mathcal{H}_S\otimes\mathcal{H}_{B,n})$ and $\tilde Z_{SB}^{n,\alpha}:=\tilde Z_S^\alpha \tilde Z_B^{n,\alpha}$, $\tilde Z_{S}^{\alpha}:= \tr[\tau_S^\alpha]$, $\tilde Z_{B}^{n,\alpha}:=\tr[\tau_{B,n}^\alpha]$. By direct evaluation of the relative entropy,
	\be 
	D(\tilde\tau_{SB}^{n,\alpha}\|(\tau_S\otimes\tau_{B,n})^\alpha/\tilde Z_{SB}^{n,\alpha})/\beta=\tr[\hat H_S \tilde\tau^{n,\alpha}_S]+\tr[\hat H_B^{(n)} \tilde\tau_{B,n}^\alpha]-(\alpha\beta)^{-1}S(\tau_S^\alpha\otimes\tau_{B,n}^\alpha/\tilde Z_{SB}^{n,\alpha})+\ln(\tilde Z_{SB}^{n,\alpha}),
	\ee
	where we have used unitary in-variance of the von Neumann entropy $S(\cdot)$.
	Thus since
	\begin{align} 
		0=& D((\tau_S\otimes\tau_{B,n})^\alpha/\tilde Z_{SB}^{n,\alpha}\|(\tau_S\otimes\tau_{B,n})^\alpha/\tilde Z_{SB}^{n,\alpha})/\beta\\
		=&\tr[\hat H_S \tau_S^\alpha/\tilde Z_{S}^\alpha]+\tr[\hat H_B^{(n)} \tau_{B,n}^\alpha/\tilde Z_{B}^{n,\alpha}]-(\alpha\beta)^{-1}S(\tau_S^\alpha\otimes\tau_{B,n}^\alpha/\tilde Z_{SB}^{n,\alpha})+\ln(\tilde Z_{SB}^{n,\alpha}),
	\end{align}
	we conclude
	\be\label{eq:relative D consequence}
	D(\tilde\tau_{SB}^{n,\alpha}\|(\tau_S\otimes\tau_{B,n})^\alpha/\tilde Z_{SB}^{n,\alpha})/\beta=\tr[\hat H_S \tilde \tau^{n,\alpha}_S]+\tr[\hat H_B^{(n)} \tilde\tau_B^{n,\alpha}]-\tr[\hat H_S \tau_S^\alpha/\tilde Z_{S}^{n,\alpha}]-\tr[\hat H_B^{(n)} \tau_{B,n}^\alpha/\tilde Z_{B}^{n,\alpha}].
	\ee
	Energy conservation implies 
	\be\label{eq:energy conservation}
	\tr[\hat H_{SB}^{(n)}(\tau_S\otimes\tau_{B,n})^\alpha/\tilde Z_{SB}^{n,\alpha}]=\tr[\hat H_{SB}^{(n)}\tilde\tau_{SB}^{n,\alpha}].
	\ee
	Combining Eqs. \eqref{eq:energy conservation}, \eqref{eq:relative D consequence} we achieve
	\be\label{eq:relative entropy trace rel}
	D(\tilde\tau_{SB}^{n,\alpha}\|(\tau_S\otimes\tau_{B,n})^\alpha/\tilde Z_{SB}^{n,\alpha})= \tr[\lambda \hat I_n(\tilde\tau_{SB}^{n,\alpha}-(\tau_S\otimes\tau_{B,n})^\alpha/\tilde Z_{SB}^{n,\alpha})]\beta.
	\ee
	Pinsker inequality states that for any two density matrices $\rho$, $\sigma$,
	\be
	D(\rho\|\sigma)\geq \frac{1}{2} \|\rho-\sigma\|_1^2.
	\ee
	It follows from it,
	and from Eq. \eqref{eq:relative entropy trace rel}, 
	\begin{align} 
		\|U_n(\tau_S\otimes\tau_{B,n})^\alpha U_n^\dag -(\tau_S\otimes\tau_{B,n})^\alpha\|_1 & \leq \tilde Z_{SB}^{n,\alpha}\beta\sqrt{2\,  \tr[\lambda\hat I_n(\tilde\tau_{SB}^{n,\alpha}-(\tau_S\otimes\tau_{B,n})^\alpha/\tilde Z_{SB}^{n,\alpha})] } 
		\,\,\\
		& \leq 2\tilde Z_{SB}^{n,\alpha}\beta\sqrt{\sup_{\rho\in\mathcal{S}(\mathcal{H}_S\otimes\mathcal{H}_{B,n})} \left |\tr[\hat I_n\,\rho]\right| \lambda }\\
		& \leq 2\tilde Z_{SB}^{n,\alpha} \beta \sqrt{\lambda \|\hat I_n\| }
	\end{align}
\end{proof}

This lemma may be of independent interest, as it makes explicit the idea mentioned in previous work such as \cite{cwiklinski2015limitations} of how Davies maps, in the weak coupling limit, can be taken as free operations in the resource theory of athermality \cite{brandao2013resource}. 

With it at hand, we can prove the central lemma.

\begin{lemma}[Lemma \ref{thorem:Li-Winter conjecture for Davies maps1} of main text]\label{thorem:Li-Winter conjecture for Davies maps2}
	Assume conditions in Section \ref{sec:Conditions for Theorem Li-winter conjec} hold. Then all maps $T_t(\cdot)$ satisfy the inequality
	\be\label{Eq:proof theorem 1 eq A}
	D((\cdot)\|\tau_S)-D(T_t(\cdot)\|\tau_S)\geq D\left((\cdot)\big{\|}\tilde T_t(T_t(\cdot))\right), \quad \forall \,t\geq 0
	\ee
	where $\tilde T_t(\cdot)$ is the Petz recovery map corresponding to $T_t(\cdot)$, 
	\begin{equation}
		\tilde{T}_t(\cdot)= \tau^{1/2}_S T^\dagger_t \left(\tau^{-1/2}_S (\cdot) \tau^{-1/2}_S\right) \tau^{1/2}_S,
	\end{equation}
	with $T_t^\dag$ denoting the adjoint of $T_t$.
\end{lemma}
\begin{proof}
	Had there been no interaction term (i.e. $\lambda=0$) and the bath been finite dimensional, the proof of this Lemma would have been straightforward, using the techniques developed in \cite{wehner2015work} involving simple manipulations of the relative entropy, and the data processing inequality for finite dimensional baths. The added difficulty here will be in proving monotone convergence as the bath Hilbert space tends to infinity. To achieve this, we will use Lemma \ref{lem:fixed pt almost comm} and continuity arguments. We will perform the calculations for the map $ \tr_B\left[ \me^{-\mi \tilde t \hat H_{SB}} (\cdot)\otimes\rho_B\, \me^{\mi \tilde t \hat H_{SB}}  \right]$ rather than $T_t(\cdot)$ itself. We will finally take the limit described in Eq. \eqref{eq:Davies map limit} to conclude the proof.\\
	Noting that the relative entropy between two copies is zero, followed by using its additivity and unitarity invariance properties, we find for $\rho_S\in\mathcal{S}(\mathcal{H}_S)$,
	\begin{align}
		D(\rho_S\|\tau_S)&=D(\rho_S\otimes\tau_{B,n}\|\tau_S\otimes\tau_{B,n})=D(U_n\rho_S\otimes\tau_{B,n} U^\dag_n\|U_n\tau_S\otimes\tau_{B,n} U^\dag_n)\\
		&= D(U_n\rho_S\otimes\tau_{B,n} U^\dag_n\|\tau_S\otimes\tau_{B,n} +\sqrt{\lambda}\hat B_n(\lambda))\label{eq:1 theorem proof}
	\end{align}
	where
	$\hat B_n(\lambda):=\left(U_n\tau_S\otimes\tau_{B,n} U^\dag_n-\tau_S\otimes\tau_{B,n} \right)/\sqrt{\lambda}$. 
	
	With the identity $D(\gamma_{CD}\|\zeta_{CD})-D(\gamma_{D}\|\zeta_{D})=D\left(\gamma_{CD}\|\exp(\ln\zeta_{CD}+\ln\id_C\otimes\zeta_D-\ln\id_C\otimes\zeta_D)\right)$ for bipartite states $\gamma_{CD}$, $\zeta_{CD}$, we have that
	\begin{align} 
		&D(U_n\rho_S\otimes\tau_{B,n} U^\dag_n\|\tau_S\otimes\tau_{B,n} +\sqrt{\lambda}\hat B_n(\lambda)) -D(\sigma_S\|\tau_S +\sqrt{\lambda}\tr_{B,n}[\hat B_n(\lambda)])\label{eq:firt line D -D =D}\\
		&=D\Bigg(U_n\rho_S\otimes\tau_{B,n} U^\dag_n\Big{ \|}\exp\Big(\ln (\tau_S\otimes\tau_{B,n}+\sqrt{\lambda}\hat B_n(\lambda))+\ln\sigma_S\otimes\id_{B,n}-\ln(\tau_S+\sqrt{\lambda}\tr_{B,n}[\hat B_n(\lambda)])\otimes\id_{B,n}\Big)\Bigg)\\
		&\geq D\Bigg(\rho_S\Big{ \|}\tr_{B,n}\bigg[U^\dag_n\exp\Big(\ln (\tau_S\otimes\tau_{B,n}+\sqrt{\lambda}\hat B_n(\lambda))+\ln\sigma_S\otimes\id_{B,n}-\ln(\tau_S+\sqrt{\lambda}\tr_{B,n}[\hat B_n(\lambda)])\otimes\id_{B,n}\Big)U_n\bigg]\Bigg)
		\label{eq:3 theorem proof}
	\end{align}
	where $\sigma_{S,n}:=\tr_{B,n}[U_n\rho_S\otimes\tau_{B,n} U^\dag_n]$ and in the last line we have used the unitarity invariance of the relative entropy followed by the data processing inequality.
	Plugging Eq. \eqref{eq:1 theorem proof} into Eq. \eqref{eq:3 theorem proof} followed by taking the $n\rightarrow\infty$ limit, we obtain 
	\begin{align}
		& D(\rho_S\|\tau_S) -D(\sigma_{S}\|\tau_S +\sqrt{\lambda}\tr_{B}[\hat B(\lambda)])\\
		& \geq D\Bigg(\rho_S\Big{ \|}\tr_{B}\bigg[U^\dag\exp\Big(\ln (\tau_S\otimes\tau_{B}+\sqrt{\lambda}\hat B(\lambda))+\ln\sigma_{S}\otimes\id_{B}-\ln(\tau_S+\sqrt{\lambda}\tr_{B}[\hat B(\lambda)])\otimes\id_{B}\Big)U\bigg]\Bigg),\label{eq:D-D>D after nlim} 
	\end{align}
	where we have defined $\hat B(\lambda):=\lim_{n\rightarrow\infty} \hat B_n(\lambda)$, $\sigma_S:=\lim_{n\rightarrow \infty}\sigma_{S,n}$. Before continuing, we will first note the validity of Eq. \eqref{eq:D-D>D after nlim}. 
	We start by showing that $\hat B(\lambda)$ is trace class for $\lambda\in[0,1]$. From Lemma \ref{lem:fixed pt almost comm} it follows
	\be 
	\|\hat B_n(\lambda)\|_1\leq 2\tilde Z_{SB}^{n,1}\beta\sqrt{\|\hat I_n\|},
	\ee
	for all $\lambda\in[0,1]$ with the r.h.s. $\lambda$ independent. By definition of $\tilde Z_{SB}^{n,\alpha}$ , it follows that it is the partition function of a tensor product of thermal states on $\mathcal S (\mathcal{H}_S\otimes\mathcal{H}_{B,n})$ at inverse temperatures $\alpha \beta_S,\alpha\beta$. Since the Hamiltonians $\hat H_{B,1},\hat H_{B,2},\hat H_{B,3},\ldots,\hat H_{B}$ by definition have well defined thermal states (finite partition functions) for all positive temperatures, it follows that $\lim_{n\rightarrow \infty} \tilde Z_{SB}^{n,\alpha}<\infty$ for all $\alpha>0$. Thus noting that by definition, $\lim_{n\rightarrow \infty} \|\hat I_n\|=\|\hat I\|$ and that $\hat I$, is a bounded operator, it follows that
	\be\label{eq:B B pre lambda lim}
	\|\hat B(\lambda)\|_1=\lim_{n\rightarrow \infty}\|\hat B_n(\lambda)\|_1=  2\tilde Z_{SB}^{\infty,1}\beta\sqrt{\|\hat I\|} <\infty,
	\ee
	Thus since $\tau_S+\sqrt{\lambda}\tr_B[\hat B(\lambda)])$ is finite dimensional and  Hermitian, and the eigenvalues of finite dimensional Hermitian matrices are continuous in their entries \cite{rellich,kato}
	, it follows, since $\tau_S$ has full support, that there exists $0<\lambda^*\leq 1$ such that for all $\lambda\in[0,\lambda^*]$, $\tau_S +\sqrt{\lambda}\tr_{B}[\hat B(\lambda)]$ has full support. Thus for all $\lambda\in[0,\lambda^*]$, the r.h.s. of Eq. \eqref{eq:D-D>D after nlim} is upper bounded by a finite quantity uniformly in $n\rightarrow\infty$ and thus since relative entropies are non-negative by definition, Eq. \eqref{eq:D-D>D after nlim}  is well defined for all $\lambda\in[0,\lambda^*]$.
	
	We now set $\Delta$ appearing in $U$ to $\Delta=t/\lambda^2$ followed by taking the limit $\lambda\rightarrow 0^+$ while keeping $t$ fixed   in Eq.  \eqref{eq:D-D>D after nlim}  thus achieving
	\be\label{eq:pre eq theorem proof}
	D(\rho_S\|\tau_S) -D(T_t(\rho_S)\|\tau_S)\geq D\Big(\rho_S\Big{ \|}\tr_{B}\big[U^\dag T_t(\rho_S)\otimes\tau_{B}U\big]\Big),
	\ee
	where we have used that by definition, $T_t(\cdot)=\lim_{\lambda\rightarrow 0^+} \tr_B[U(\cdot)\otimes\tau_B U^\dag].$

	We now proceed to calculate the Petz's recovery map for the map $T_t(\cdot)$. The adjoint map is $\tr_{B}[\tau_{B}^{1/2}U^\dag((\cdot)\otimes\id_{B})U\tau_{B}^{1/2}]$. Hence from the definition in Eq.\eqref{eq:Petz1} it follows that the Petz recovery map for $T_t(\cdot)$ is
	\be\label{eq: T tilde def}
	\tilde T_t(\cdot):=\tau_S^{1/2}\tr_{B}\left[\tau_{B}^{1/2}U^\dag\left(\tau_{S}^{-1/2}(\cdot)\tau_{S}^{-1/2}\otimes\id_{B}\right)U\tau_{B}^{1/2}\right]\tau_S^{1/2}.
	\ee
	Similarly to before, we define a traceless, self adjoint operator $\tilde B=\tilde B(\lambda):=\left( U\tau_S^{1/2}\otimes\tau_{B}^{1/2}U^\dag-\tau_S^{1/2}\otimes\tau_{B}^{1/2}\right)/\sqrt{\lambda}$. In analogy with the reasoning which led to Eq. \eqref{eq:B B pre lambda lim}, it follows from Lemma \ref{lem:fixed pt almost comm} that $\|\tilde B(\lambda)\|_1=\lim_{n\rightarrow \infty}\|\tilde B_n(\lambda)\|_1=  2\tilde Z_{SB}^{\infty,1/2}\beta\sqrt{\|\hat I\|} <\infty,$ for all $\lambda\in[0,1]$. For general $U=\exp(-\mi \Delta \hat H_{SB})$, we can now write
	\begin{align}
		&\tau_S^{1/2}\tr_{B}\left[\tau_{B}^{1/2}U^\dag\left(\tau_{S}^{-1/2}(\cdot)\tau_{S}^{-1/2}\otimes\id_{B}\right)U\tau_{B}^{1/2}\right]\tau_S^{1/2}\\
		&=\tr_{B}\left[\left(U^\dag\tau_S^{1/2}\otimes\tau_{B}^{1/2}+\sqrt{\lambda} U^\dag\tilde B\right)\left(\tau_{S}^{-1/2}(\cdot)\tau_{S}^{-1/2}\otimes\id_{B}\right)\left( \tau^{1/2}_{S}\otimes\tau_B^{1/2} U+\sqrt{\lambda}\tilde B U\right)\right]\\
		&=\tr_{B}[U^\dag\left((\cdot)\otimes\tau_{B}\right) U]+\sqrt{\lambda}\hat g_1(\cdot)+\lambda\hat g_2(\cdot)\in\mathcal{S}(\mathcal{H}_B),\label{eq:tilde M in proof}
	\end{align}
	where
	\begin{align}\label{eq:tilde T as U dag U}
		\hat g_1(\cdot) &= \tr_{B}\left[U^\dag \tilde B(\tau_S^{-1/2}(\cdot)\otimes\tau_{B}^{1/2})U \right] +\tr_{B}\left[U^\dag ((\cdot)\tau_S^{-1/2}\otimes\tau_{B}^{1/2})\tilde B U \right]  \\
		\hat g_2(\cdot) &= \tr_{B}\left[ U^\dag\tilde B\left(   \tau_S^{-1/2}(\cdot)\tau_S^{-1/2} \otimes \id_{B}\right)\tilde B U  \right],
	\end{align}
	which are well defined since they are comprised of products of bounded operators. Similarly to before, in Eq. \eqref{eq:tilde M in proof} we now set $\Delta$ appearing in $U$ to $\Delta=t/\lambda^2$ followed by taking the limit $\lambda\rightarrow 0^+$ while keeping $t$ fixed achieving
	\be 
	\tilde T_t(\cdot)\\
	=\tr_{B}[U^\dag\left((\cdot)\otimes\tau_{B}\right) U]
	\ee
	where we have used Eq. \eqref{eq: T tilde def}. Hence substituting Eq. \eqref{eq:tilde T as U dag U} in to Eq. \eqref{eq:pre eq theorem proof} and noting the equations holds for all states $\rho_S\in\mathcal{S}(\mathcal{H}_B)$, we conclude the proof.
\end{proof}
\begin{remark}
	In the above proof, we have taken two independent limits, namely 1st the infinite bath volume limit ($n\rightarrow\infty$) followed by the Van Hove limit ($\lambda\rightarrow 0^+$ while keeping $t$ fixed). This is the order in which Davies performed the limits  \cite{davies1974,davies1979generators} when defining the Davies map. From physical reasoning, one would expect the Davies map to be equally valid if the order of the limits is reversed. We note that the proof of Theorem \ref{thorem:Li-Winter conjecture for Davies maps} follows also if the order of the these two limits is reversed but now with the new definitions
	\begin{align}
		T_t(\cdot)&=\lim_{n\rightarrow\infty} \lim_{\lambda\rightarrow 0^+} \tr_{B,n}\left[ U_n(\tilde t) (\cdot)\otimes\tau_{B,n}\, U_n^\dag(\tilde t)  \right]\in\mathcal{S}(\mathcal{H}_S) \quad \text{subject to } \tilde t \lambda^2=t  \text{ fixed.}\label{eq:reverse D map}\\
		\tilde T_t(\cdot)&=\lim_{n\rightarrow\infty} \lim_{\lambda\rightarrow 0^+} \tau_S^{1/2}\tr_{B,n}\left[\tau_{B,n}^{1/2}U_n^\dag(\tilde t)\left(\tau_{S}^{-1/2}(\cdot)\tau_{S}^{-1/2}\otimes\id_{B,n}\right)U_n(\tilde t)\tau_{B,n}^{1/2}\right]\tau_S^{1/2}.\in\mathcal{S}(\mathcal{H}_S) \quad \text{subject to } \tilde t \lambda^2=t  \text{ fixed.}\label{eq:reverse tilde D map}
	\end{align}
	An interesting technical question is whether the above limits commute i.e. whether Eqs. \eqref{eq:reverse D map}, \eqref{eq:reverse tilde D map} are identical to Eqs. \eqref{eq:Davies map limit}, \eqref{eq: T tilde def}.
\end{remark}

\subsection{Quantum detailed balance and Petz recovery map}\label{sec:QDB}

Now we show that all Davies maps have the peculiar property that they are the same as their Petz recovery map. This is because of a crucial property satisfied by their generators: \emph{quantum detailed balance}. For Theorem \ref{thorem:Li-Winter conjecture for Davies maps2} in the main text to hold, we require both the conditions of Section \ref{sec:Conditions for Theorem Li-winter conjec} and the following Lemma to hold. For the sake of generality, we show the that the results is true for any fixed point $\Omega$ with full support. We remind the reader, that a dynamical semigroup $M_t(\cdot)$ is a one parametre family of CPTP maps with a generator consisting of a unitary part $	\theta(\cdot)=-[\hat H_\text{eff},\cdot],$ and a dissipative part called a Lindbladian, $\mathcal{L}(\cdot)$, such that all together we have
\be 
M_t(\cdot)=\me^{t\mi\theta+t\mathcal{L}}(\cdot)
\ee.

\begin{theorem}[Dissipative Recovery map]\label{lem:qdb}
	A quantum dynamical semigroup $M_t(\cdot)$ with no unitary part, $\theta=0$, and Lindbladian $\mathcal{L}$ satisfying quantum detailed balance (Eq. \eqref{eq:QDB}) for the state $\Omega$ with full rank, is equal to its corresponding Petz recovery map, namely,
	\begin{equation}
		{ M}_t(\cdot)=\tilde{ M}_t(\cdot),
	\end{equation}
	where
	\begin{equation}
		\tilde{ M}_t(\cdot) =  \Omega^{1/2} M_t^\dagger(M_t(\Omega)^{-1/2} \cdot M_t(\Omega)^{-1/2}) \Omega^{1/2}.
	\end{equation}
	
\end{theorem}
\begin{proof}
	
	The property of quantum detailed balance (also sometimes referred to as the reversibility, or KMS condition) reads 
	\begin{equation}\label{eq:qdb}
		\langle A, \mathcal{L}^\dagger(B) \rangle_{\Omega} = \langle \mathcal{L}^\dagger(A), B \rangle_{\Omega}
	\end{equation}
	for all $A,B\in\cc^{d_S\times d_S}$, where $\mathcal{L}^\dagger$ is the adjoint Lindbladian, and we define the scalar product
	
	\begin{equation} \label{eq:scalarp}
		\langle A, B \rangle_{\Omega} := \text{Tr}[\Omega^{1/2} A^\dagger \Omega^{1/2} B].
	\end{equation}
	
	Because Eq. \eqref{eq:qdb} holds for all $A,B\in\cc^{d_S\times d_S}$, Eq. \eqref{eq:qdb} implies that \cite{ohya2004quantum}
	
	\begin{equation}\label{eq:qdbid}
		\mathcal{L}(\cdot)=\Omega^{1/2}  \mathcal{L}^\dagger (\Omega^{-1/2} \cdot \Omega^{-1/2}) \Omega^{1/2}.
	\end{equation}
	
	Eq. \eqref{eq:qdb} automatically implies that any power of the generator also obeys the same relation, that is, $\forall \,n \in \nn^+$
	
	\begin{align} 
		\langle A, \mathcal{L}^{\dagger n}(B) \rangle_{\Omega} &= \langle A, \Omega^{-1/2}  \mathcal{L} (\Omega^{1/2}\ldots \Omega^{-1/2}  \mathcal{L} (\Omega^{1/2} B \Omega^{1/2}) \Omega^{-1/2} \ldots\Omega^{1/2}) \Omega^{-1/2} \rangle_{\Omega} 
		\\
		&=\langle A, \Omega^{-1/2}  \mathcal{L}^{ n} (\Omega^{1/2} B \Omega^{1/2}) \Omega^{-1/2} \rangle_{\Omega} 
		\\
		&= \langle \mathcal{L}^{\dagger \,n}(A), B \rangle_{\Omega},\label{eq:qdbn}
	\end{align}
	where in the first line we use Eq. \eqref{eq:qdbid} $n$ times and the 2nd line follows from the definition of the adjoint map. Hence we can also write 
	\begin{equation}\label{eq:qdbidn}
		\mathcal{L}^n(\cdot)=\Omega^{1/2}  \mathcal{L}^{\dagger \,n} (\Omega^{-1/2} \cdot \Omega^{-1/2}) \Omega^{1/2}.
	\end{equation}

	The semigroup can be written as $M_t(\cdot)=e^{\mathcal{L}t}(\cdot)$. Its adjoint semigroup is given by $e^{\mathcal{L}^\dagger t}$ and hence the Petz recovery map is (see Eq.\eqref{eq:Petz1})
	\begin{equation}\label{eq:pets dyamic recov map}
		\tilde{ M}_t(\cdot) =  \Omega^{1/2} e^{\mathcal{L}^\dagger t}(\Omega^{-1/2} \cdot \Omega^{-1/2}) \Omega^{1/2}.
	\end{equation}
	Since, $\tilde M_t(\cdot)=\Omega^{1/2}\left(\,\sum_{n=0}^\infty (t\mathcal{L})^{\dag n}(\Omega^{-1/2} \cdot \Omega^{-1/2})/(n!) \,\right)\Omega^{1/2}$, Eq. \eqref{eq:pets dyamic recov map} together with Eq. \eqref{eq:qdbidn}, means that $\tilde{M}_t(\cdot)=M_t (\cdot)$.
	
	We note that the Petz recovery map is defined in terms of a map $\Gamma(\cdot)$ and a state $\sigma_S$ as the unique solution to 
	
	\begin{equation}\label{eq:Petz1}
		\langle A,\Gamma^\dagger(B) \rangle_{\sigma_S} = \langle \tilde \Gamma^\dagger (A), B \rangle_{\Gamma (\sigma_S)}
	\end{equation}
	for all $A,B\in\cc^{d_S\times d_S}$ and the scalar product is given by Eq. \eqref{eq:scalarp}. The solution takes the form \cite{ohya2004quantum}
	
	\begin{equation}
		\tilde{ T}(\cdot) =  \sigma_S^{1/2} T^\dagger(T(\sigma_S)^{-1/2} \cdot T(\sigma_S)^{-1/2}) \sigma_S^{1/2},
	\end{equation}
	such that we always have that $\tilde{T}(T(\sigma_S))=\sigma_S$. Here this simplifies by choosing $\sigma_S=\Omega$ a fixed point of $M_t(\cdot)$.
\end{proof}

When the generator is time-independent and $\theta=0$, we thus have from Theorem \ref{lem:qdb} that the combination of a map for a time $t$ and its recovery map is equivalent to applying the map for a time $2t$. That is $\tilde M_t(M_t(\cdot))=M_{2t}(\cdot)$.
This means we can write Eq. \eqref{Eq:proof theorem 1 eq A} in a particularly simple form.\\
The following Lemma builds on Theorem \ref{lem:qdb} to extend it to the case in which the dynamical semigroup also includes a unitary part.
\begin{lemma}[Dissipative and unitary Recovery map]\label{lem:qdb unitaary}
	Let $M_t(\cdot)$ be a quantum dynamical semigroup with unitary part  $\theta$ and Lindbladian $\mathcal{L}$ which: 1) satisfying quantum detailed balance (Eq. \eqref{eq:QDB}) for the state $\Omega$ with full rank and 2) commute, $\theta(\mathcal{L} (\cdot)) =\mathcal{L} (\theta(\cdot))$. Then, $M_t(\cdot)$ has a Petz recovery map $\tilde M(\cdot)$ which is a dynamical semigroup with unitary part $-\theta$ and Lindbladian $\mathcal{L}$. Namely, if 
	\begin{equation}
		{ M}_t(\cdot)=\me^{t \mi \theta+t\mathcal{L}}(\cdot),
	\end{equation}
	satysfying 1) and 2), then
	\begin{equation}
		\tilde{ M}_t(\cdot) 
		=\me^{-t \mi \theta+t\mathcal{L}}(\cdot).
	\end{equation}
\end{lemma}
\begin{proof}
	We just need to note two facts:
	\begin{itemize}
		\item The Petz recovery map of a unitary map $U(\cdot)U^\dag$ that had fixed point $\Omega$ is $U^\dag(\cdot)U$.
		\item The Petz recovery map of a composition of two maps with the same fixed point, is equal to the composition of the Petz recovery maps of the individual maps, i.e. $\widetilde{\Gamma_1\circ\Gamma_2}=\tilde\Gamma_2\circ\tilde\Gamma_1$ (this is one of the key properties listed in \cite{li2014squashed}).
	\end{itemize} 
	We hence can write the recovery map of $M_t(\cdot)$ as
	\begin{equation}
		\tilde M_t(\rho_S)=e^{t\mathcal{L}}( e^{i H_{\text{eff}} t} \rho_S e^{-i H_{\text{eff}} t})=e^{i H_{\text{eff}} t} e^{t\mathcal{L}}(\rho_S) e^{-i H_{\text{eff}} t}.
	\end{equation}
	The only difference between $M_t$ and $\tilde M_t$ is the change of sign in the time of the unitary evolution. The recovery map is then made up of the dissipative part evolving forwards, and the unitary part evolving backwards in time.
\end{proof}

\begin{theorem}(Theorem \ref{thorem:Li-Winter conjecture for Davies maps} of main text)\label{thorem:twice the evolved time}
	Assume conditions in Section \ref{sec:Conditions for Theorem Li-winter conjec} hold and $T_t(\cdot)$ satisfies quantum detailed balance (Eq. \eqref{eq:QDB}) and has zero unitary part, $\theta=0$. Then $T_t(\cdot)$ 
	satisfies the inequality
	\be\label{Eq:proof theorem 1 eq2}
	D((\cdot)\|\tau_S)-D(T_t(\cdot)\|\tau_S)\geq D\left((\cdot)\big{\|} T_{2t}(\cdot))\right),\quad t\geq 0.
	\ee
\end{theorem}
\begin{proof}
	Direct consequence of Theorems \ref{lem:qdb}, \ref {thorem:Li-Winter conjecture for Davies maps2}.
\end{proof}
\begin{remark}[When $\theta\neq 0$]
	Due to properties 3), 5) of the main text satisfied by Davies maps, and the unitary invariance of the Relative entropy (i.e. $D(U\cdot U^\dagger\|U\cdot U^\dagger)=D(\cdot\|\cdot)$), it follows
	\be 
	D((\cdot)\|\tau_S)-D(T_t(\cdot)\|\tau_S)=D((\cdot)\|\tau_S)-D(\me^{t\mathcal{L}}(\cdot)\|\tau_S),
	\ee
	and thus the l.h.s. of Eq. \eqref{Eq:proof theorem 1 eq2} is the same even when a non zero unitary part is included.
	Furthermore, we note that the canonical form of Daives maps have  $\theta(\mathcal{L} (\cdot)) =\mathcal{L} (\theta(\cdot))$ by definition [see property 3) in main text] and thus, due to Lemma \ref{lem:qdb unitaary}, even when $\theta\neq 0$, we have that
	\be 
	D((\cdot)\|\tilde T_t(T_t(\cdot)))=D((\cdot)\|\me^{2t\mathcal{L}}(\cdot)),
	\ee
	which is the r.h.s. of Eq. \eqref{Eq:proof theorem 1 eq2}.  Thus applying Theorem \ref{thorem:Li-Winter conjecture for Davies maps}, we have
	\be 
	D((\cdot)\|\tau_S)-D(T_t(\cdot)\|\tau_S)=D((\cdot)\|\tau_S)-D(\me^{t\mathcal{L}}(\cdot)\|\tau_S)\geq D((\cdot)\|\me^{2t\mathcal{L}}(\cdot)),
	\ee
	for any $\theta$.
\end{remark}

\subsection{Spohn's inequality: rate of entropy production}

\label{Spohn}

We give an alternative proof of a well-known result which was first shown in \cite{spohn1978entropy} that gives the expression for the infinitesimal rate of entropy production of a Davies map. This is stated without a proof in many standard references such as \cite{breuer2002theory,alicki2007quantum}. Then we show in a similar way how in the infinitesimal time limit our lower bound becomes trivial.

First we need the following lemma, which proof can be found in, for instance, \cite{carlen2010trace}.
\begin{lemma} \label{le:adler} Let $\id\in\cc^{n\times n}$ be the identity matrix, and $A,B\in\cc ^{n\times n}$ be matrices such that both $A$ and $A+tB$ are positive with $t\in\rr$, we have that
	\begin{equation}\label{eq:adler}
		\log{(A+tB)}-\log{A}=t \int_0^1 \frac{1}{(1-x)A+x \id}B\frac{1}{(1-x)A+x \id} \text{d}x+\mathcal{O}(t^2)
	\end{equation}
\end{lemma}

With this, we can show the following:
\begin{theorem}\label{lemma3}
	Let $\mathcal{L}(\rho_S (t))$ be the generator of a dynamical semigroup, with a fixed point $\tau_S$ such that $\mathcal{L}(\tau_S)=0$. We have that the entropy production rate $\sigma(\rho_S (t))$ is given by
	\begin{equation}\label{eq:entpinf}
		\sigma(\rho_S (t)):= -\frac{ \text{d}D(\rho_S (t) || \tau_S)}{\text{d} t}=\text{Tr}[\mathcal{L}(\rho_S (t))(\log{\tau_S}-\log{\rho_S (t)})] + \text{Tr}[\mathcal{L}(\rho_S (t)) \Pi_{\rho_S (t)}]\ge 0,
	\end{equation}
	where $\Pi_{\rho_S (t)}$ is the projector onto the support of $\rho_S (t)$. The second term of the sum vanishes at all times for which the rate is finite.
\end{theorem}
\begin{proof}
	The last inequality (positivity) follows from the data processing inequality for the relative entropy, so we only need to prove the equality. The proof only requires Lemma \ref{le:adler} and some algebraic manipulations. We have that
	\begin{align}
		&\frac{\text{d} D(\rho_S (t) || \tau_S)}{\text{d}t}= \lim_{h\rightarrow 0} \frac{D(\rho_{t+h} || \tau_S)-D(\rho_S (t) || \tau_S)}{h}  \\
		&=\lim_{h\rightarrow 0} \frac{\text{Tr}[(\rho_S (t)+\mathcal{L}(\rho_S (t))h) (\log{\{\rho_S (t)+\mathcal{L}(\rho_S (t))h\}}-\log{\tau_S})]-\text{Tr}[\rho_S (t) (\log{\rho_S (t)}-\log{\tau_S})]}{h} \\
		&=\lim_{h\rightarrow 0} \frac{1}{h}\Big[ \text{Tr}[(\rho_S (t)+\mathcal{L}(\rho_S (t))h) \{\log({\rho_S (t)})+h \int_0^1 \frac{1}{(1-x)\rho_S (t)+x \id}\mathcal{L}(\rho_S (t))\frac{1}{(1-x)\rho_S (t)+x \id} \text{d}x\}-\log{\tau_S})]  \nonumber
		\\ & \qquad \qquad \qquad \qquad \qquad \qquad \qquad \qquad \qquad \qquad \qquad \qquad \qquad \qquad \qquad -\text{Tr}[\rho_S (t) (\log{\rho_S (t)}-\log{\tau_S})] \Big] \\
		&= \text{Tr}[\mathcal{L}(\rho_S (t))(\log{\rho_S (t)}-\log{\tau_S})] + \text{Tr}[\rho_S (t) \int_0^1 \frac{1}{(1-x)\rho_S (t)+x \id}\mathcal{L}(\rho_S (t))\frac{1}{(1-x)\rho_S (t)+x \id} \text{d}x ] \\ \label{eq:2term}
		&= \text{Tr}[\mathcal{L}(\rho_S (t))(\log{\rho_S (t)}-\log{\tau_S})] + \text{Tr}[\rho_S (t) \mathcal{L}(\rho_S (t)) \int_0^1 \Big(\frac{1}{(1-x)\rho_S (t)+x \id}\Big)^2 \text{d}x ] .
	\end{align}
	Where to go from the 2nd to the third line we used Lemma \ref{le:adler}, and from the 4th to the 5th we use the ciclicity and linearity of the trace. Now note the following integral
	\begin{equation}
		\int_0^1 \Big(\frac{1}{(1-x)p+x}\Big)^2 \text{d}x=\frac{1}{p} \,\,\, \forall \, p \neq 0.
	\end{equation}
	This means that, on the support of $\rho_S (t)$, 
	\begin{equation}\label{eq:supp}
		\int_0^1 \Big(\frac{1}{(1-x)\rho_S (t)+x \id}\Big)^2 \text{d}x=\frac{1}{\rho_S (t)}.
	\end{equation}
	Note that outside the support of $\rho_S (t)$ this integral is not well defined. Given this, we can write
	\begin{equation} \label{eq:entp}
		\frac{\text{d} D(\rho_S (t) || \tau_S)}{\text{d}t}=\text{Tr}[\mathcal{L}(\rho_S (t))(\log{\rho_S (t)}-\log{\tau_S})] + \text{Tr}[\mathcal{L}(\rho_S (t)) \Pi_{\rho_S (t)}],
	\end{equation}
	where $\Pi_{\rho_S (t)}$ is the projector onto the support of $\rho_S (t)$.
	The Lindbladian is traceless $\text{Tr}[\mathcal{L}(\rho_S (t))]=0$ and hence second term of this Equation vanishes as long as $\text{supp} (\mathcal{L}(\rho_S (t))) \subseteq \text{supp}(\rho_S (t))$, which we can expect for most times. At instants in time when this is not the case and this term may give a finite contribution (that is, when the rank increases), the first term in Eq. \eqref{eq:entp} diverges logarithmically \cite{spohn1978entropy}, and hence that finite contribution is negligible.
\end{proof}
A similar reasoning can be used to show that the instantaneous lower bound on entropy production  rate that we can get from our main result in Eq. \eqref{Eq:proof theorem} is trivial for most times. In particular, we can show
\begin{lemma} \label{result2}
	The lower bound of Eq. \eqref{Eq:proof theorem} vanishes in the limit of infinitesimal time transformations. More precisely, we have that
	\begin{equation}
		\lim_{h \rightarrow 0} \frac{D(\rho_S (t) || \rho_S({t+2h}))}{h} = -2 \text{Tr} [\mathcal{L}(\rho_S (t)) \Pi_{\rho_S (t)}],
	\end{equation}
	where $\Pi_{\rho_S (t)}$ is the projector onto the support of $\rho_S (t)$.
	This vanishes as long as $\text{supp} (\mathcal{L}(\rho_S (t))) \subseteq \text{supp}(\rho_S (t))$.
\end{lemma}
\begin{proof}
	The proof is similar to the one for Theorem  \ref{lemma3} above.
	\begin{align}
		\lim_{h \rightarrow 0} \frac{D(\rho_S (t) || \rho({t+2h}))}{h} &= \lim_{h \rightarrow 0} \frac{1}{h} \text{Tr}[\rho_S (t) (\log{\rho_S (t)}-\log{(\rho_S (t)+2h \mathcal{L}(\rho_S (t)))} ] \\
		&=\text{Tr}[-2 \rho_S (t) \int_0^1 \frac{1}{(1-x)\rho_S (t)+x \id}\mathcal{L}(\rho_S (t))\frac{1}{(1-x)\rho_S (t)+x \id} \text{d}x\}] \\
		&= -2 \text{Tr} [\mathcal{L}(\rho_S (t)) \Pi_{\rho_S (t)}],
	\end{align}
	where in the second line we applied Lemma \ref{le:adler}, and in the third we used Eq. \eqref{eq:supp}.
\end{proof}
Hence for infinitesimal times, the lower bound gives the same condition as the positivity condition in Eq. \eqref{eq:entpinf}. It will be nonzero only when $\text{supp} (\mathcal{L}(\rho_S (t))) \nsubseteq \text{supp}(\rho_S (t))$, in which case the rate of entropy production diverges (at points in time when the rank of the system increases).

\section{Proof of Theorem \ref{lemmak}}
\label{app:lemmak}

Here we prove the following theorem from the main text: 

\begin{theorem} 
	[Tightness of the entropy production bound] The largest constant $k\geq 0$ such that 
	\begin{equation}\label{eq:largek2}
		F_\beta (\rho_S(0))-F_\beta (\rho_S(t)) \geq \frac{1}{\beta}D\left(\rho_S(0) \big{\|} \rho_S(k \,t)\right)
	\end{equation}
	holds for all Davies maps, is $k=2$.
\end{theorem}
\begin{proof}
	We show the inequality is violated for any $k>2$ by finding a simple family of Davies maps for which the violation is proven analytically.	
	
	Let us take the general form of a Davies map on a qubit, and act on a state with initial density matrix $\rho$ without coherence in energy\footnote{We assume no coherence for simplicity. An analogous, yet longer, proof of the violation of inequality Eq. \eqref{eq:largek2} for $k>0$ holds for the case of coherence in energy is possible.}  $\rho(0)=\text{diag} (p(0), 1-p(0)) $, and with a corresponding thermal state $\tau=\text{diag}(q,1-q)$. 
	The time evolution of the Davies map is only that of the populations (as no coherence in the energy eigenbasis is created), and it takes the general form
	\begin{equation} \label{eq:timee}
		\begin{pmatrix}
			p(t) \\ 1-p(t)
		\end{pmatrix}
		=
		\begin{pmatrix}
			1-a_t && a_t\frac{q}{1-q} \\
			a_t && 1-a_t\frac{q}{1-q}
		\end{pmatrix}
		\begin{pmatrix} 
			p(0) \\ 1-p(0)
		\end{pmatrix},
	\end{equation}
	where $a_t=(1-q)(1-e^{-A t})$ for some $A>0$.
	Let us now define the function
	\begin{equation}
		g(t,k) := \beta F_\beta (\rho_S(0))- \beta F_\beta (\rho_S(t)) - D\left(\rho_S(0) \big{\|} \rho_S(k \,t)\right),
	\end{equation}
	and the variable $x:=e^{-A t}$. One can show, after some algebra, that for the time evolution of Eq.\eqref{eq:timee}
	\begin{align}
		g(x,k) =\big[ (q-1)+x(p(0)-q)\big]\log{\big(1+x\frac{q-p(0)}{1-q}\big)}-\big[ q+x(p(0)-q) \big] \\ \nonumber
		+(1-p(0))\log{\big( 1+x^k \frac{q-p(0)}{1-q} \big)}+p(0) \log {\big( 1+ x^k \frac{p(0)-q}{q}\big)}.
	\end{align}
	For large $t$, $x$ will be arbitrarily small and hence we can expand the logarithms up to leading order in $x$. The zeroth and first order terms in $x$ cancel out, and we obtain
	\begin{equation}
		g(x,k) = \frac{-1}{2q(1-q)} x^2 (p(0)-q)^2+\frac{1}{q(1-q)}x^k (p(0)-q)^2 + \mathcal{O}(x^3).
	\end{equation}
	We see that if $k>2$, for sufficiently large time, the $k$-th order term will be very small compared to the $x^2$ one, which is always negative. For $k=2$ we have 
	\begin{equation}
		g(x,2) = \frac{1}{2q(1-q)} x^2 (p(0)-q)^2 + \mathcal{O}(x^3),
	\end{equation}
	such that the leading order is always positive.
	This completes the proof. 
	
\end{proof}

\section{Maps Beyond Davies}\label{sec: appendix Maps Beyond Davies}
Given that the inequality in Eq. \eqref{Eq:proof theorem 1 eq} is saturated in some limits, such as when the evolution approaches the fixed point, it is unlikely that a stronger inequality of a similar kind can be derived even in particular cases. However, general results are known for CPTP maps, leading to weaker forms of such bound. In this Section we state the best known general result from \cite{junge2015universal} and show how they simplify in particular cases of maps with properties similar to Davies maps. This means that we can also bound the entropy production of maps that may not be Davies maps.

The result, which proof involves techniques from complex interpolation theory, is the following:

\begin{theorem}
	(Main result of \cite{junge2015universal})
	Let $\Gamma(\cdot)$ be a CPTP map, and $\rho$, $\sigma$ any two quantum states. We have that
	\begin{equation}\label{eq:junge}
		D(\rho || \sigma) - D(\Gamma(\rho) || \Gamma(\sigma)) \ge -2 \int_\mathbb{R} \text{d}t\,\, p(t) \log{F(\rho, \tilde{\Gamma}_t (\Gamma(\rho))))},
	\end{equation}
	where $F(\rho,\sigma)=\tr[\sqrt{\sqrt{\sigma}\rho \sqrt{\sigma}}]$ is the quantum fidelity, the map $\tilde{\Gamma}_t$ is the \emph{rotated} recovery map
	\begin{equation}
		\tilde{\Gamma}_t(\cdot)=\sigma^{i t} \tilde \Gamma (\Gamma(\sigma)^{-i t} \cdot \Gamma(\sigma)^{i t})\sigma^{-i t}
	\end{equation}
	and $p(t)$ is the probability density function $p(t)=\frac{\pi}{2} (\cosh (\pi t)+1)^{-1}$.
\end{theorem}
\begin{proof}
	See \cite{junge2015universal}.
\end{proof}
We now observe that the rotated map can be simplified given the following conditions:
\begin{itemize}
	\item If the map has a fixed point $\Gamma(\Omega)=\Omega$, the Petz recovery map simplifies 
	to become
	\begin{equation}
		\tilde{\Gamma}_t(\cdot)=\Omega^{i t} \tilde \Gamma (\Omega^{-i t} \cdot \Omega^{i t})\Omega^{-i t} \,\, \forall t \in \mathbb{R}
	\end{equation} 
	This by itself implies that $\tilde{\Gamma}_t(\Omega)=\Omega$.
	\item The map may also obey the property of \emph{time-translation symmetry}, where this is given by
	\begin{equation}\label{eq:tts app}
		\Gamma(\cdot)=\Omega^{i t} \Gamma(\Omega^{-i t} \cdot \Omega^{i t}) \Omega^{-i t}.
	\end{equation}
	
	If a map obeys this symmetry, the adjoint map $\Gamma^\dagger (\cdot)$ also will. This can be seen through the definition of the adjoint, which is that for any two matrices $A,B$,
	\begin{equation}
		\text{Tr}[A \Gamma (B)]=\text{Tr}[\Gamma^\dagger (A) B],
	\end{equation}
	and in particular, it holds for the matrices $A'=\Omega^{i t} A \Omega^{-i t}$, $B'=\Omega^{-i t} B \Omega^{i t}$. This, together with Eq. \eqref{eq:tts app}, means that 
	\begin{equation}
		\text{Tr}[\Gamma^\dagger (A) B]=\text{Tr}[\Omega^{i t} \Gamma^\dagger(\Omega^{-i t} \cdot \Omega^{i t}) \Omega^{-i t} (A) B]
	\end{equation}
	
	Hence this property, together with the fixed-point property, means that the rotated recovery map becomes equal to the Petz map, and the integral in Eq. \eqref{eq:junge} gets averaged out. 
	
	It may be the case, however, that the symmetry exists, but that the fixed point is not the thermal state, and hence the simplification does not occur. This may be the case for instance when there is weak coupling to a non-thermal environment.
	
	\item If on top of these two conditions the map has the property of quantum detailed balance, namely
	\be 
	\langle A, \Gamma^\dagger (B) \rangle_{\Omega} = \langle \Gamma^\dagger (A), B \rangle_{\Omega}, \text{ 
		for all } A,B\in\cc^{d_S\times d_S},
	\ee
	the Petz recovery map and the original one are the same $\tilde \Gamma (\cdot)=\Gamma(\cdot)$. Examples of maps which satisfy detailed balance which are not Davies maps do exist. See \cite{Mwolf_Cubitt} for general characterization of quantum dynamical semigroups. 
\end{itemize}

When all these hold we have that Eq. \eqref{eq:junge} becomes
\begin{equation}\label{eq:junge2}
	D(\rho || \Omega) - D(\Gamma(\rho) || \Omega) \ge -2 \log{F(\rho, \Gamma (\Gamma(\rho)))}.
\end{equation}
This bound could be tightened by replacing the $-2 \log{F(\rho, \Gamma (\Gamma(\rho)))}$ with the measured relative entropy, $D_\mathbb{M}(\rho\| \Gamma (\Gamma(\rho))))$ \cite{sutter2016multivariate}. This would achieve a tighter bound, although at the expense of it being less explicit, unless one could solve the maximization problem in the definition of the measured relative entropy. If the map is a dynamical semigroup with a time-independent generator $\Gamma=M_t$, we may also write $M_t(M_t( \cdot))=M_{2t}(\cdot)$.

Davies maps have all these properties. Further examples where all these properties appear are semigroups derived from the low-density limit (which models a system immersed in an ideal gas at low density, see \cite{alicki2007quantum} for details), or the so-called \emph{heat bath generators} \cite{kastoryano2014quantum}.

We note however that $D( \rho || \sigma) \ge -2 \log{F(\rho,\sigma)}$, and hence Eq. \eqref{eq:junge2} is a weaker bound than Eq. \eqref{Eq:proof theorem 1 eq}, and in particular is not tight as the fixed point is approached. 

\section{Equivalence of definitions of quantum detailed balance}\label{app:eqqdb}

In the literature, different nonequivalent definitions of the property of quantum detailed balance have been given. 
While in many places the one given is that of Eq. \eqref{eq:QDB}, an alternative definition, which can be found for instance in \cite{alicki2007quantum,kossakowski1977} is that the Lindbladian is self-adjoint with respect to the inner product
\begin{equation}\label{eq:qdb v2}
	\langle A, \mathcal{L}^\dagger(B) \rangle'_{\Omega} = \langle \mathcal{L}^\dagger(A), B \rangle'_{\Omega},
\end{equation}
for all $A,B\in\cc^{d_S\times d_S}$, where the inner product is defined as
\begin{equation}\label{eq:new inner prod}
	\langle A, B \rangle'_\Omega=\tr{[\Omega A^\dagger B]}.
\end{equation}
Eq. \ref{eq:new inner prod} is different from that of Eq. \eqref{eq:scalarp} due to the noncommutativity of the operators. The solution to Eq \eqref{eq:qdb v2} is \cite{G_rule}
\begin{equation}\label{eq:qdb2}
	\mathcal{L}(\cdot)=\Omega \mathcal{L}^\dagger (\Omega^{-1} \cdot),
\end{equation}
while the solution to Eq. \eqref{eq:QDB} is \cite{ohya2004quantum}
\begin{equation}\label{eq:qdbid 3}
	\mathcal{L}(\cdot)=\Omega^{1/2}  \mathcal{L}^\dagger (\Omega^{-1/2} \cdot \Omega^{-1/2}) \Omega^{1/2}.
\end{equation}
We now give a simple proof of the fact that, under the condition that the map is time-translation invariant w.r.t. fixed point, the two conditions are the same.

\begin{theorem}
	For a Lindbladian operator $\mathcal{L}(\cdot)$ which obeys the property of time-translation symmetry w.r.t. fixed point $\Omega$ of full Rank (Eq. \eqref{eq:tts}), the quantum detailed balance conditions of Eq. \eqref{eq:qdb2} and Eq. \eqref{eq:qdbid 3} are equivalent.
\end{theorem}

\begin{proof}
	We rewrite both Eq. \eqref{eq:qdb2} and Eq. \eqref{eq:qdbid 3} in terms of their individual matrix elements in the orthonormal basis $\{\ket{i}\}$ in which $\Omega=\sum_i p_i \ket{i}\bra{i}$ is diagonal. Eq. \eqref{eq:qdb2} can be written in the form
	\begin{equation}\label{me1}
		\bra{i} \mathcal{L}(\ket{n}\bra{m})\ket{j}= \frac{p_i}{p_n}\bra{i} \mathcal{L}^\dagger(\ket{n}\bra{m})\ket{j}
	\end{equation}
	and Eq. \eqref{eq:qdbid 3} is
	\begin{equation}\label{me2}
		\bra{i} \mathcal{L}(\ket{n}\bra{m})\ket{j}= \sqrt{\frac{p_i p_j}{p_n p_m}} \bra{i} \mathcal{L}^\dagger(\ket{n}\bra{m})\ket{j}.
	\end{equation}
	We can see that for each matrix element the conditions only change by the factors multiplying in front, which are different unless $\frac{p_n}{p_m}=\frac{p_i}{p_j}$.
	
	Let us now introduce the following decomposition of operators in $\cc^{d_S\times d_S}$ in terms of their \emph{modes of coherence}
	\begin{equation}
		A=\sum_\omega A_\omega,
	\end{equation}
	where $A_\omega$ is defined as
	\begin{equation}
		A_\omega=\sum_{\substack{k,l \\ \text{s.t. }\omega=\log{\frac{p_k}{p_l}}}} \proj{k} A \proj{l}.
	\end{equation}
	The name of  \emph{modes of coherence} is due to the fact that under the action of the unitary $\Omega^{-i t} \cdot \Omega^{i t}$ they rotate with a different Bohr frequency, that is
	\begin{equation}
		\Omega^{-i t} A_\omega \Omega^{i t}= A_\omega e^{-i \omega  t}.
	\end{equation}
	
	If the Lindbladian has the property of time-translational invariance w.r.t. the fixed point (Eq. \eqref{eq:tts}), it can be shown \cite{marvian2014modes,lostaglio2015quantum} that each input mode is mapped to its corresponding output mode of the same Bohr frequency $\omega$. We can write this fact as
	\begin{equation}
		\mathcal{L}(A_\omega) = \mathcal{L}(A)_\omega.
	\end{equation}
	This means that in Eq. \eqref{me1} and \eqref{me2}, $\bra{i} \mathcal{L}(\ket{n}\bra{m})\ket{j}=0 $ unless the Bohr frequencies coincide at the input and the output, that is, when $\log{\frac{p_n}{p_m}}=\log{\frac{p_i}{p_j}}$. That is, the two conditions are nontrivial only in those particular matrix elements in which both are equivalent.
\end{proof}

\end{document}